\documentclass[a4paper, amsfonts, amssymb, amsmath, reprint, showkeys,  twoside, onecolumn,notitlepage,nofootinbib]{revtex4-1}

\def\apj{{ApJ}}                 

\def\mnras{{MNRAS}}             
\def\prc{{Phys.~Rev.~C}}        
\def\prd{{Phys.~Rev.~D}}        
\def\prl{{Phys.~Rev.~Lett.}}    
\def\nat{{Nature}}              

\usepackage{amsmath,amstext}
\usepackage[T1]{fontenc}
\usepackage{amssymb}
\input{epsf}
\usepackage{graphicx}
\usepackage{ae,aecompl}

\usepackage{hyperref}
\usepackage{amsmath}
\usepackage{amssymb}
\usepackage{mathtools}
\usepackage{bm}
\usepackage{cleveref}
\usepackage{tensor}
\usepackage{braket}
\usepackage{enumitem}
\usepackage{mhchem}
\usepackage{amsthm}

\usepackage{color}

\newcommand{\spc}{\quad \quad \quad}

\def\be{\begin{equation}}
\def\ee{\end{equation}}
\def\beq{\begin{eqnarray}}
\def\eeq{\end{eqnarray}}

\theoremstyle{definition}
\newtheorem{definition}{Definition}

\theoremstyle{theorem}
\newtheorem{theorem}{Theorem}

\begin{document}
\title{Can we make sense of dissipation without causality?}
\author{L.~Gavassino}
\affiliation{Nicolaus Copernicus Astronomical Center, Polish Academy of Sciences, ul. Bartycka 18, 00-716 Warsaw, Poland}

\begin{abstract}
Relativity opens the door to a counter-intuitive fact: a state can be stable to perturbations in one frame of reference, and unstable in another one. For this reason, the job of testing the stability of states that are not Lorentz-invariant can be very cumbersome. We show that two observers can disagree on whether a state is stable or unstable only if the perturbations can exit the light-cone. Furthermore, we show that, if a perturbation exits the light-cone, and its intensity changes with time, due to dissipation, then there are always two observers that disagree on the stability of the state. Hence, ``stability'' is a Lorentz-invariant property of dissipative theories if and only if the principle of causality is respected. We present 14 applications to physical problems from all areas of relativistic physics, ranging from theory to simulation.
\end{abstract}

\maketitle

\section{Introduction}


Deterministic field theories (such as hydrodynamics, classical electrodynamics and general relativity) find application in all areas of physics, ranging from condensed matter physics to string theory. Recently, the whole area of classical field theory is receiving a new boost, due to the experimental advances, such as the discovery of the Quark-Gluon Plasma at RHIC and LHC \cite{QGPreview2017} and the now common-place detection of GW-mergers from compact objects by LIGO, Virgo and KAGRA \cite{Liguz2019}, which have driven the development of an ever increasing number of fluid-like theories, to describe exotic phenomena of all kinds \cite{Heller2014,Sadoogi2018,Florskoski2019,
GavassinoKhalatnikov2021,GavassinoQuasiHydro2022}. Most notably, relativistic dissipative hydrodynamics is becoming a standard tool in the study of a host of physical problems, from high-energy physics \cite{FlorkowskiReview2018} to astrophysics \cite{Shibatuz2017,CamelioBulk1_2022,CamelioBulk2_2022}. 

The search for the ``correct'' field theory for describing a given phenomenon typically involves formulating a large number of alternative candidate theories, many of which are then ruled out, or proven to be equivalent to others. Usually, there is so much freedom in the construction of a phenomenological theory, that it is easy to get lost in the landscape of alternative formulations. For example, there are at least 11 different formulations of relativistic viscous hydrodynamics  \cite{Eckart40,landau6,Israel_Stewart_1979,LindblomRelaxation1996,Liu1986,
carter1991,BemficaDNDefinitivo2020,BaierRom2008,
Denicol2012Boltzmann,SrickerOttinger2019,VanStableFirst2012}, 7 formulations of superfluid hydrodynamics
\cite{khalatnikov_book,prix2004,cool1995,lebedev1982,
Son2001,koby2018PhRvC,Gusakov08}, and 6 formulations of radiation hydrodynamics
\cite{Thomas1930,Weinberg1971,UdeyIsrael1982,
AnileRadiazion1992,GavassinoRadiazione,Farris2008}. However, in a relativistic setting, all this freedom comes at a price: most of the theories that one can formulate lead to completely unphysical predictions \cite{Hiscock_Insatibility_first_order}. For example, since flow of energy equals density of momentum, in some (unphysical) theories, a fluid can spontaneously accelerate, departing form equilibrium, and pushing heat in the opposite direction to conserve the total momentum \cite{Hishcock1988,GavassinoLyapunov_2020}. Pathologies of this kind constitute a serious problem for numerical simulations, because unphysical artefacts cannot be separated from physical effects.

Luckily, there is a standard procedure that allows us to the test the reliability of a relativistic theory and rule out a considerable fraction of candidate theories: the causality-stability assessment. The idea is simple: a theory can be considered reliable only if signals do not propagate faster than light (causality\footnote{{In the theory of relativity, the word ``causality'' stands for ``subluminal propagation of information'' \cite{Hawking1973,Wald,BemficaCausality2018,
Susskind1969,Fox1970}. This concept was introduced because the additional term ``$\, vx_B \,$'' in the relativistic transformation of time, $t_A=\gamma (t_B + v x_B)$,  can push a future event to the past (or vice versa), provided that $x_B^2 > t_B^2$. Hence, if information could propagate faster than light, there would be an observer for whom it is propagating towards the past. To avoid grandfather-like paradoxes, it was conjectured that superluminal communication, just like communication to the past, should be impossible, because the effect must always follow the cause (hence the term causality). Logically speaking, this reasoning is not very rigorous \cite{kessence}. But it worked! The principle of causality is built into the mathematical structure of the Standard Model of particle physics \cite{Peskin_book,Eberhard1988,Keister1996} and, to date, it has never been falsified.}}), and if the state of thermodynamic equilibrium (or the vacuum, for zero-temperature theories) is stable against (possibly large\footnote{Throughout the article, we use the generic word ``perturbation'' as a synonym of ``disturbance'', namely an alteration (i.e. displacement) of a region of the medium from its equilibrium state. According to this terminology, a perturbation is not necessarily small, unless we say it explicitly. In general, the thermodynamic equilibrium state should be stable against all kinds of perturbations, both small and large \cite{Hishcock1988}, although in most situations one is able to rigorously asses stability only in the linear regime.}) perturbations. Since decades, there is a whole line of research devoted to assessing these two properties 
\cite{Stuekelberg1962,
Hishcock1983,OlsonLifsh1990,GerochLindblom1990,
Geroch_Lindblom_1991_causal,Pu2010,
Kovtun2019,Lopez11,Bemfica2019_conformal1,
BrutoThird2021,GavassinoGibbs2021,
GavassinoCausality2021}. Unfortunately, the assessment procedure is complicated (especially for what concerns stability), and the proposed theories are much more numerous than those that are, then, effectively tested. It is clear that a universal and easily applicable criterion, that can be used to quickly asses if a theory is stable or not, would be a breakthrough for the field (which is exactly what this paper provides).   

One aspect of the assessment is particularly problematic. When we study the dynamics of small perturbations around the vacuum state, the linearised field equations are the same in all reference frames, because we are linearising a Lorentz-convariant theory around a Lorentz-invariant state. On the other hand, if the unperturbed state has finite temperature (or chemical potential), its total four-momentum defines a preferred reference frame, so that the linearised field equations look different in different frames. This opens the doors to a counter-intuitive fact: at finite temperature, the equilibrium state may be stable in one reference frame, but unstable in another one. This paradox is possible only because different observers impose their initial data on different constant-time hypersurfaces, and hence deal with different initial-value problems \cite{Kost2000,GavassinoLyapunov_2020,GavassinoUEIT2021}. The result is that one needs to test the stability of the equilibrium in all reference frames, to be sure that a theory really makes sense. This is unfortunate, because the stability analysis in a reference frame in which the system is moving can be very cumbersome (the background is anisotropic).

The goal of this paper is to finally resolve the paradox of systems that are stable in one reference frame and unstable in others. We will prove that this can happen \textit{only if} the principle of causality is violated. The intuition behind this fact is that two observers can disagree on whether a perturbation is growing or decaying only if (by relativity of simultaneity \cite{special_in_gen}) the perturbation can be chronologically reordered, so that the two observers disagree on which part of the perturbation is in the past, and which is in the future. Since this can happen only if the perturbation propagates outside the light-cone, it follows that you need to violate causality, if you want to have two observers disagreeing on a stability assessment. This simple idea, once formulated in mathematical terms, will result into two theorems, according to which causal theories that are stable in one reference frame are also stable in any other frame.

In the following, I will first describe the physical setup of the problem, and the general physical mechanisms at the origin of the instability of dissipative theories. I will then rigorously prove the main result in section \ref{causalstablerelazione}. A reader not interested in the technical details may, however, skip to section \ref{unishiipunti}, where I will provide a simple argument that summarizes the essence of the whole paper, or directly to section \ref{apliacia}, when I will present 14 examples of concrete applications of our results to theories that are commonly used in a number of fields.

In case some readers wish to have a brief summary of how the relativistic stability assessment usually works, they can see Appendix \ref{AAAAAAAAA} for a quick overview.
Particular emphasis is given to the differences with the non-relativistic case. The mathematical and logical foundations of the method were laid in \cite{Hiscock_Insatibility_first_order}.

Throughout the paper we adopt the signature $ ( - , +, + , + ) $ and work in natural units $c=1$. The space-time is Minkowski, with metric $g$; we use global inertial coordinates, generically denoted by $x^a$ (so that $\nabla_a=\partial_a$). Finally, all observers are inertial observers, i.e. they do not accelerate and they do not rotate.

\section{Some pertinent context}

The idea that there could be a connection between causality violations and instabilities has a long history, which may be summarised in the words of \citet{Israel_2009_inbook}: ``\textit{If the source of an effect can be delayed, it should be possible for a system to borrow energy from its ground state, and this implies instability}''. This argument is a restatement of the Hawking-Ellis vacuum conservation theorem \cite{Hawking1973}, according to which, if energy can enter an empty region faster than the speed of light, then the dominant energy condition is violated, and the energy density may become negative in some reference frame. Unfortunately, these ideas are not applicable to our case, because we are not studying the stability of the vacuum state, but that of a finite-temperature equilibrium state. More importantly, causality violations can occur even in systems that obey the dominant energy condition. For example, take a barotropic perfect fluid with equation of state\footnote{The reader should not be concerned about the fact that $dP/d\rho <0$ (thermodynamic inconsistency) for some $\rho$: in our proof of principle, we only need an acausal field theory, well-defined for any $\rho \geq 0 $, with smooth coefficients in the field equations, and such that $\rho >|P|$.}
\begin{equation}\label{prruz}
P(\rho) = \dfrac{\rho}{3} \big[1+ \sin(\rho^2) \big],
\end{equation}
where $P$ is the pressure and $\rho$ is the energy density, in some fixed units. This fluid is consistent with the dominant energy condition ($\rho>|P|$), but its equations are acausal, because the speed of sound $dP/d\rho$ is unbounded above.

Luckily, it is not so hard to modify the idea of Israel, adapting it to our case of interest: we only need to replace ``energy'' with ``entropy'' and ``ground state'' with ``equilibrium state'' \cite{GavassinoCausality2021}. Let us see in more detail how this works with a simple qualitative argument.

\subsection{Acausality + Dissipation = Instability?}\label{ilprimoluiluilui}

Imagine that a signal travels between two events $p$ and $q$, which are space-like separated, i.e. $g(p-q,p-q)>0$. By relativity of simultaneity \cite{special_in_gen}, we know that there are some reference frames in which $p$ happens before $q$, and other reference frames in which $q$ happens before $p$. Hence, in some reference frames the signal is travelling superluminally from $p$  to $q$, while in other frames it travels superluminally from $q$ to $p$.

Now, imagine to repeat this experiment, placing between $p$ and $q$ a dissipative medium, which absorbs the signal along the way. Then, the signal is emitted from, say, $p$. It travels in the direction of $q$, but it decays before reaching $q$. But in those reference frames in which $q$ happens before $p$, we observe that the signal is spontaneously generated in the middle of the medium, it grows without any external influence (nothing happens at $q$), and travels to $p$. Thus, the medium is unstable to the spontaneous generation of perturbations! One may argue that this type of perturbation is not really spontaneous, because still we need an emitter/receiver at $p$ for it to occur. However, the argument still works if we send $p$ at space-like infinity, so that we are left with a medium that absorbs/emits a space-like beam, which travels from/to infinity.  

The idea of the argument above is the same as that of \citet{Israel_2009_inbook}: if the cause of a signal (i.e. $p$) can be delayed, then the system can spontaneously generate a perturbation, borrowing entropy from the equilibrium state, and reversing the dissipative processes that should, instead, damp the perturbation. This implies instability.

Besides this qualitative argument, what are the concrete indications that causality and instabilities may be related? Let us have a brief summary of the present understanding of the causality-stability problem.

\subsection{Breakdown of causality and stability in infrared theories}

\begin{figure}
\begin{center}
\includegraphics[width=0.5\textwidth]{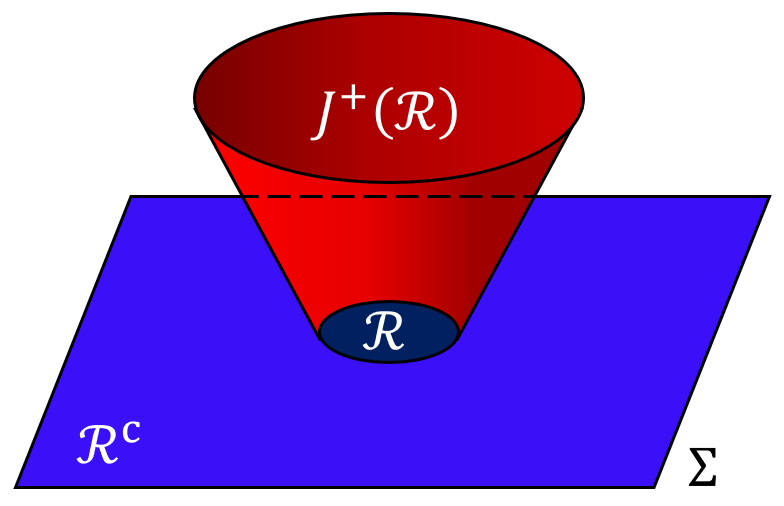}
	\caption{ Geometric visualization of the principle of causality. Take an arbitrary spacelike Cauchy 3D-surface $\Sigma$. The simplest example of such a surface is the hyperplane $\{ t=0 \}$. Divide $\Sigma$ into two regions: $\mathcal{R}$ (dark blue) and $\mathcal{R}^c$ (light blue). ``Paint in red'' all the timelike and lightlike curves that originate from $\mathcal{R}$ and propagate towards the future. The red paint will cover a set $J^+(\mathcal{R})$, called ``domain of influence of $\mathcal{R}$'', or ``causal future of $\mathcal{R}$'', or ``future light-cone of $\mathcal{R}$''.  Causality demands the following: if we compare two arbitrary solutions (of the field equations) whose initial data differ on $\mathcal{R}$, but coincide on $\mathcal{R}^c$, such solutions can differ only inside $J^+(\mathcal{R})$. This is equivalent to saying that information coming from $\mathcal{R}$ can \textit{never} exit $J^+(\mathcal{R})$. }
	\label{fig:JDR}
	\end{center}
\end{figure}

For deterministic field theories, the principle of causality reduces to a mathematical condition on the field equations: a variation of the initial data in a region of space $\mathcal{R}$ cannot affect the solution outside the future light-cone of $\mathcal{R}$ \cite{Hawking1973,Wald,BemficaCausality2018}, see figure \ref{fig:JDR}. If the equations are linear, causality also means that the retarded Green's function has support within the future light-cone \cite{Susskind1969,Fox1970}. It turns out that many phenomenological equations in physics are not consistent with this
causality criterion and, therefore, allow for super-luminal propagation of signals. The best known example is the diffusion equation $\partial_t T =D  \partial_x^2 T$, whose Green's function is
\begin{equation}\label{greenfunctionheat}
\mathcal{G}(t,x)= \dfrac{1}{\sqrt{4\pi Dt}} \exp\bigg( -\dfrac{x^2}{4Dt} \bigg) \, ,
\end{equation}
whose tails extend far beyond the future light-cone, propagating energy and information at infinite speeds. Such unphysical violations of the principle of causality usually occur in theories that are the low-frequency limit of some ``more complete'' causal theories \cite{Weymann1967,LindblomRelaxation1996,GavassinoUEIT2021}. This is, indeed, the case of the diffusion equation, that is (at least in ideal gases \cite{Israel_Stewart_1979,Denicol2012Boltzmann,Denicol_Relaxation_2011}) the low frequency limit of the telegraph equation \cite{cattaneo1958,Jou_Extended,rezzolla_book}, which is known to be causal. The same is true for the Schr{\"o}dinger equation, which is the acausal low frequency limit of the Klein-Gordon equation (with the field redefinition $\phi=e^{-imt}\psi$ \cite{Zee2003}).
\begin{equation}\label{super!}
\begin{matrix*}[l]
\text{Telegraph:} \\
\\
\text{Klein-Gordon:} \\
\end{matrix*}
\quad 
\begin{matrix*}[l]
(\tau \partial_t^2  +\partial_t) T = D \partial_x^2 T  \\
\\
(\partial_t^2  +m^2)\phi =\partial_x^2 \phi \\
\end{matrix*}
 \spc \xrightarrow{\partial_t \rightarrow 0} \spc 
 \begin{matrix*}[l]
\text{Diffusion:} \\
\\
\text{Schr{\"o}dinger:} \\
\end{matrix*}
\quad
\begin{matrix*}[l]
\partial_t T =D \partial_x^2 T \\
\\
i\partial_t \psi = -   \partial_x^2 \psi/2m  \\
\end{matrix*}
\end{equation}
For the reason above, causality violations usually occur only on very short time-scales, where the predictions of the acausal equation differ from those of its causal progenitor. In other words, causality violations usually happen outside the regime of validity of the ``infrared approximation'', upon which the acausal equation is built. Hence, one may argue that, as long as we manage to keep the high-frequency part of the solutions small, the predictions of the acausal equation should be reliable, and the causality violations negligible \cite{Fichera1992,Day1997,PanB2008,Auriault2017}.

Unfortunately, in a relativistically covariant context, keeping the acausal high-frequency part of the solutions small is almost impossible (at least in some reference frames), if the equation is acausal \textit{and} dissipative. The first authors who noticed this issue were \citet{Hiscock_Insatibility_first_order}, who verified that any Fick-type diffusion law becomes unstable in some reference frame, due to the fast growth of some unphysical high-frequency modes (see appendix \ref{AAAAAAAAA} for a quick overview of their methodology). A similar mechanism has been observed in several other systems of equations  \cite{Hishcock1983,OlsonLifsh1990,GerochLindblom1990,
Geroch_Lindblom_1991_causal,Pu2010,Kovtun2019}: 
if causality is violated, and the system is dissipative, there is some reference frame in which the system becomes unstable, due to the appearance of fast-growing modes. 


The fact that these instabilities usually depend on the frame of reference (i.e. the growing modes exist in some reference frames, but not in others) is deeply counterintuitive. Hence, it seemed natural to regard the unphysical growing modes as a mere ``mathematical pathology'' of the equations. Indeed, acausal field equations often do not present a good Cauchy problem for arbitrary data on space-like 3D-surfaces \cite{Susskind1969}; hence, it is not surprising that there is some reference frame in which an acausal theory ``misbehaves'' \cite{BaierRom2008}. However, this does not explain why dissipative systems are so exceptionally problematic: while non-dissipative acausal theories (like that considered by \citet{Susskind1969}) are singular only when the initial data is imposed on a characteristic surface, dissipative acausal systems are usually unstable in a continuum of reference frames \cite{Hiscock_Insatibility_first_order}. Hence, one may wonder whether acausality and dissipation are fundamentally incompatible. This is what we aim to understand here.

\section{Causality-Stability relations}\label{causalstablerelazione}

We have finally reached the central part of the paper. This section is arranged into three subsections, each one of which is a separate, stand-alone, result. In particular:
\begin{enumerate}
\item In subsection \ref{thetought} we present a more rigorous version of the argument given in subsection \ref{ilprimoluiluilui}, according to which, if a system is acausal and dissipative, then there is a reference frame in which it is unstable. Although linearity of the equations is never invoked explicitly, this argument is expected to be particularly useful for linear stability analyses (we also provide a concrete example in the Supplementary Material).
\item In subsection \ref{sez3} we present the following theorem: if a localised deviation from equilibrium decays over time uniformly in one reference frame, and its support does not exit the light cone, then it decays over time in all reference frames. This theorem is valid for both linear and non-linear field equations.
\item In subsection \ref{SiNNuoz} we present another theorem: if (in the linear regime) a causal theory predicts the existence of a growing sinusoidal plane-wave solution in one reference frame, then this theory is linearly unstable in all reference frames.
\end{enumerate} 
Combined together, these results should lead us to a simple stability criterion: \textit{a dissipative theory which is stable in one reference frame is causal if and only if it is stable in all reference frames}. Note that this ``causality-stability relation'' is strongly corroborated by all the explicit stability analyses that have been performed till now (which the author is aware of) for many different theories, including the Israel-Stewart theory (both in the Eckart \cite{Hishcock1983} and in the Landau \cite{OlsonLifsh1990} flow-frame), divergence-type theories \cite{GerochLindblom1990}, Geroch-Lindblom theories \cite{Geroch_Lindblom_1991_causal}, inviscid theories for heat conduction \cite{OlsonRegularCarter1990}, first-order viscous hydrodynamics \cite{Kovtun2019,BemficaDNDefinitivo2020}, second-order viscous hydrodynamics \cite{Pu2010}, third-order viscous hydrodynamics \cite{BrutoThird2021}, and Carter's multifluid theory \cite{GavassinoStabilityCarter2022}.

Since the three arguments presented in this section are stand-alone, in each subsection we will work under slightly different assumptions (e.g. in subsection \ref{sez3} we deal with non-linear deviations from equilibrium with compact support, whereas in subsection \ref{SiNNuoz} we study a linear plane wave with infinite support). However, there are three fundamental ideas that remain the same across the whole paper:
\begin{itemize}
\item ``Causality''$\, = \,$information cannot exit the light-cone \cite{Hawking1973,Wald,BemficaCausality2018,Susskind1969,Fox1970};
\item ``Instability''$\, = \,$there is a reference frame in which deviations from equilibrium can grow in time \cite{Hiscock_Insatibility_first_order};
\item ``Dissipation''$\, = \,$there is a reference frame in which deviations from equilibrium decay in time.
\end{itemize}
Eventually, this will allow us to construct a simple ``unified argument'' (in section \ref{unishiipunti}), which combines together the three main results of this section.

\subsection{Acausal dissipative systems are not covariantly stable}\label{thetought}

We consider a small perturbation that is travelling super-luminally across a medium, disturbing the equilibrium state and violating causality. We assume that such perturbation can be modelled as a localised wave-packet (like a sound pulse), which moves along a space-like world-line. If the wave-packet is highly-oscillating (ultra-violet limit), such world-line is a characteristic of the field equations. Let us also assume that there is an observer $A$ (say, Alice), in whose reference frame the system exhibits a dissipative behaviour. Since the unperturbed state is the equilibrium state, a reasonable definition of ``dissipative behaviour'' is that all localized perturbations eventually decay to zero for large times. Hence, we can require that, in the reference frame of Alice, the intensity of the perturbation is a decreasing function of time. The Minkowski diagram of this process is presented in figure \ref{fig:fig} (left panel). 

Now we immediately see the problem: since the perturbation is travelling along a space-like path, which part of this path happens ``earlier'' and which happens ``later'' depends on the frame of reference. Hence, we can surely find a second observer $B$ (say, Bob), in motion with respect to Alice, in whose reference frame the perturbation is growing in time (figure \ref{fig:fig}, right panel). Let us show it analytically.

\begin{figure}
\begin{center}
\includegraphics[width=0.7\textwidth]{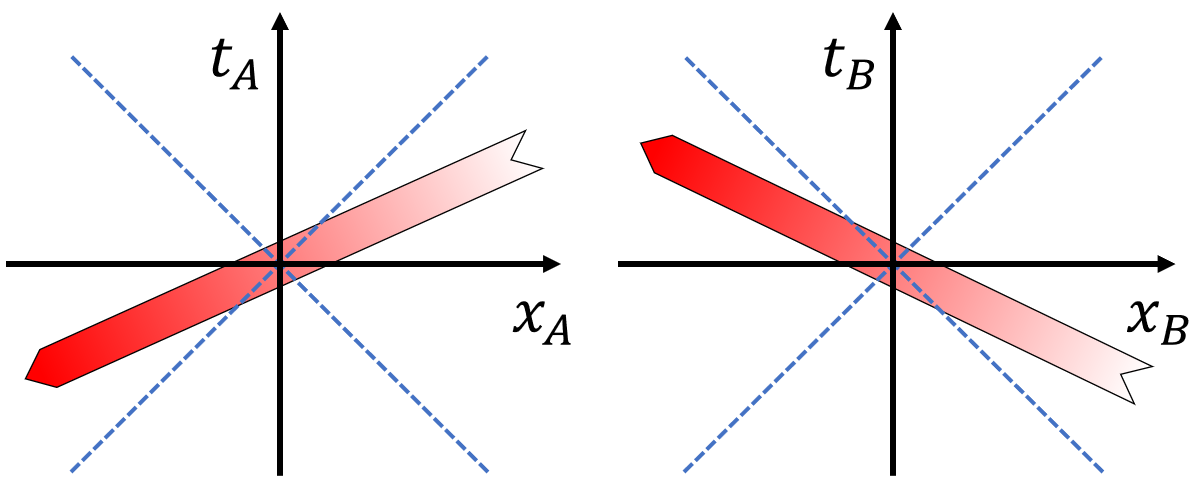}
	\caption{Minkowski diagrams of the argument outlined in section \ref{thetought}. Reference frame of Alice (left panel): the perturbation moves super-luminally from the left to the right and its intensity decreases with time as a result of dissipation. Reference frame of Bob (right panel): the perturbation moves from the right to the left and its intensity grows with time. The two points of view are connected by a Lorentz boost. The shades of red are a color-map of the intensity of the perturbation (red large, white small); the arrows have the orientation induced by $\varphi$ (see the main text); the blue dashed lines are the light-cone.}
	\label{fig:fig}
	\end{center}
\end{figure}

At each point $p$ of the space-like world-line drawn by the center of the wave-packet, we may quantify the intensity of the perturbation using a Lorentz scalar $\varphi(p)$.\footnote{For example, if $T^{ab}$ is the stress-energy tensor, $\varphi(p)$ may be the typical deviation from equilibrium (averaged over the local oscillations) of the scalar field $T^{ab}T_{ab}$, in a neighbourhood of $p$. One may also take its square, to make sure that $\varphi$ is always non-negative and plays the role of a sort of ``norm'' of the perturbation.} The inverse of the relation $\varphi(p)$ defines a Lorentz-invariant parametrization on the world-line: $p(\varphi)$. Using this parametrization, and approximating the world-line to a straight line passing through the origin, we can write a relation of the form $x_A(\varphi) = w \, t_A(\varphi)$, with $w>1$ (space-like condition). If we boost this relation to Bob's frame, we obtain
\begin{equation}\label{varfuz}
t_B(\varphi) =\gamma \, (1-vw) \, t_A(\varphi) \, ,
\end{equation}
where $v$ and $\gamma$ are the boost's velocity and Lorentz factor. Taking the derivative of \eqref{varfuz}, and inverting the result, we find
\begin{equation}\label{rubbuz}
\dfrac{d\varphi}{dt_B} = \dfrac{1}{\gamma(1-vw)} \, \dfrac{d\varphi}{d t_A} \, .
\end{equation} 
We see that, if $w^{-1}<v<1$, then the sign of $d\varphi/dt_B$ is opposite to that of $d\varphi/dt_A$. Thus, if the perturbation is damped in the reference frame of Alice ($d\varphi/dt_A <0$), it grows in the reference frame of Bob ($d\varphi/dt_B >0$), meaning that the equilibrium state is unstable in Bob's frame.

We can draw several conclusions from the argument above. First of all, we see that the instability can occur only if the system is both acausal and dissipative. In fact, if it were causal, then $w \leq 1$, and the factor $1-vw$ would always be positive; if it were non-dissipative, then $\varphi=\text{const}$, and equation \eqref{rubbuz} would reduce to the identity $0=0$. It is also immediately explained why the reference frames in which the system is unstable form a continuum: they are all those reference frames in which the chronological order of the events inside the perturbation is inverted, with respect to the chronological order perceived by Alice. Finally, by looking at equation \eqref{rubbuz}, we see that the instability is most violent close to $v=w^{-1}$, namely at the unstable-to-stable transition frame, where one has $d\varphi /dt_B = \infty$. This is a well-known feature of this kind of instabilities: rather than the growth rate, it is the growth time (the inverse of the rate!) that changes sign smoothly as we move from an unstable to a stable frame of reference \cite{Hiscock_Insatibility_first_order,GavassinoLyapunov_2020,GavassinoUEIT2021}.
In the Supplementary Material, we apply this argument to the super-luminal telegraph equation, showing that one can correctly predict the onset and the quantitative aspects of the instability without performing the whole stability analysis explicitly.

We can also make some additional comments:
\begin{itemize}
\item When $v>w^{-1}$, the perturbation grows with time in Bob's frame ($d\varphi/dt_B>0$); hence, we may say that the system looks ``anti-dissipative'' in Bob's frame. On the other hand, the obedience to the second law of thermodynamics ($\nabla_a s^a \geq 0$, where $s^a$ is the entropy current) is a Lorentz-invariant property of the system. This implies that the entropy grows also in the reference frame of Bob ($dS_B/dt_B \geq 0$). It follows that, in Bob's frame, the entropy is an increasing function of the intensity of the perturbation:
\begin{equation}
\dfrac{dS_B}{d\varphi} = \dfrac{dS_B}{dt_B} \, \dfrac{dt_B}{d\varphi} \geq 0 \, .
\end{equation}
In other words, the equilibrium state is not the maximum entropy state in Bob's frame\footnote{The frame-dependence of the maximum entropy state is not in contradiction with the Lorentz-invariance of the entropy \cite{GavassinoLorentzInvariance2021}, because the total entropy is Lorentz-invariant only at equilibrium \cite{Becattini2016}. Indeed, it is easy to see from figure \ref{fig:fig} (considering that $\nabla_a s^a \neq 0$ along the red arrows) that any attempt to use the Gauss theorem to prove that $S_A =S_B$ is doomed to fail.}. The recently discovered connection between instabilities and violations of the maximum entropy principle \cite{GavassinoLyapunov_2020,GavassinoGibbs2021,GavassinoCausality2021}
can be understood in the light of this simple argument.
\item It is evident from figure \ref{fig:fig} that, for the argument to be rigorous, the whole shape of the perturbation, and not just its peak, must be drifting super-luminally. Hence, our argument cannot be extended to causal systems whose group velocity happens to be super-luminal for some specific frequency (like those studied in \cite{Guoy1996,WangSuper2000,Withayachumnankul2010}, which can be stable \cite{Pu2010}). Only genuinely acausal systems \cite{Susskind1969} are affected by the present instability mechanism.
\item Since the high-frequency wave-packets travel on the ``acoustic cone'' (a.k.a. characteristic cone) of the field equations \cite{kessence}, we can conclude that the instability appears whenever the hyperplane $ \{ t_B=\text{const} \}$ is more sloping than the acoustic cone, so that a part of the future acoustic cone deeps below the hyperplane. Therefore, if the material is isotropic in the reference frame of Alice, the acausal dissipative theory is unstable in Bob's frame if the hyperplane $\{ t_B = \text{const} \}$ is ``time-like'' with respect to the acoustic metric
\begin{equation}
\tilde{g}^{ab} = g^{ab} + (1-w^{2})u_A^a u_A^b \spc (u_A^a = \text{Alice's four-velocity}).
\end{equation}
We will explore this point in more detail in section \ref{unishiipunti}.
\item The instability mechanism described here differs profoundly from the condensation instability of the tachyon field. In fact, the tachyon field is a causal system \cite{Susskind1969}, which is unstable in all reference frames, whereas here we are dealing with acausal systems, which are stable in some reference frames and unstable in others.
\end{itemize}

\subsection{Lorentz-invariance of dissipation}\label{sez3}

We have seen that causality violations lead to instabilities. Now we will prove that frame-dependent instabilities (namely, deviations from equilibrium that grow in Bob's frame while they decay in Alice's frame) are forbidden, if the principle of causality is respected. In this section, we will focus our attention on a localised (possibly large) ``perturbation'', namely a compactly-supported deviation of the hydrodynamic fields from their equilibrium value.

Take an arbitrary space-like Cauchy 3D-surface $\Sigma$, and decompose it into two regions $\mathcal{R}$ and $\mathcal{R}^c$, such that
\begin{equation}
\mathcal{R} \cup \mathcal{R}^c = \Sigma   \quad \spc \mathcal{R} \cap \mathcal{R}^c = \emptyset \spc \quad \mathcal{R} \text{ is compact} .
\end{equation}
Using $\Sigma$ as the initial-data hypersurface, suppose that there is an initial (linear or non-linear) displacement from equilibrium, confined within $\mathcal{R}$. This is what we mean by ``localised perturbation''. Physically, such perturbation can be any kind of non-equilibrium phenomenon, like a hot spot, a soliton, a vortex ring, a chemical imbalance, or even an ``explosion'' (in $\mathcal{R}$). We construct a non-negative scalar field $\varphi$, which measures how far the system is from equilibrium at each spacetime event, and vanishes wherever the perturbation is absent (hence $\varphi =0$ on $\mathcal{R}^c$). If the theory is well-behaving, such ``perturbation-intensity field'' (namely, $\varphi$) can always be constructed, see Appendix \ref{appendoxB} (a rigorous mathematical definition of ``perturbation'' is provided in Appendix \ref{techno}). The following definition is natural \cite{Hawking1973,Wald,BemficaCausality2018}:
\begin{definition}[sub-luminality]\label{sub}
The perturbation is \textit{sub-luminal} if $\varphi(p) =0$ for any event $p \in \mathcal{D}^+(\mathcal{R}^c)$, the future Cauchy development of $\mathcal{R}^c$.  
\end{definition}
\noindent An equivalent definition of sub-luminality is that $\varphi \neq 0$ only on $J^+(\mathcal{R})$ (the causal future of $\mathcal{R}$), see figure \ref{fig:fig2}, left panel. Now, if $u_A^a$ is Alice's four-velocity, we can define Alice's time-coordinate in a Lorentz-covariant fashion:
\begin{equation}
t_A = -x_a u_A^a  \, .
\end{equation}
Hence, interpreting $t_A$ as a scalar field, we can define the sets
\begin{equation}
J^+_A (t):= \{ \, \text{events }p \, | \, t_A(p) \geq t \, \} \, .
\end{equation}
Each set $J^+_A (t)$ is simply the causal future of the hyperplane $t_A=t$. Then, we can make a second definition:
\begin{definition}[dissipation]\label{diss}
A sub-luminal perturbation is dissipated in the reference frame of Alice if, $\forall \, \varepsilon >0$, there exists $t_\varepsilon \in \mathbb{R}$ such that $\varphi(p)<\varepsilon$ for any event $p \in J^+(\mathcal{R}) \cap J^+_A (t_\varepsilon)$. 
\end{definition}
This is a condition of uniform convergence of the perturbation to zero: after a certain time $t_\varepsilon$ (in Alice's rest frame), the intensity of the perturbation falls below $\varepsilon$ everywhere, and stays below $\varepsilon$ for $t_A \geq t_\varepsilon$ (see shades of red in figure \ref{fig:fig2}, left panel). Think of $\varepsilon$ as the instrumental resolution: at $t_\varepsilon$, the system is back in equilibrium within resolution $\varepsilon$. Analogous definitions can be made for Bob: just replace $A$ with $B$. We can finally present our theorem:

\begin{figure}
\begin{center}
\includegraphics[width=0.5\textwidth]{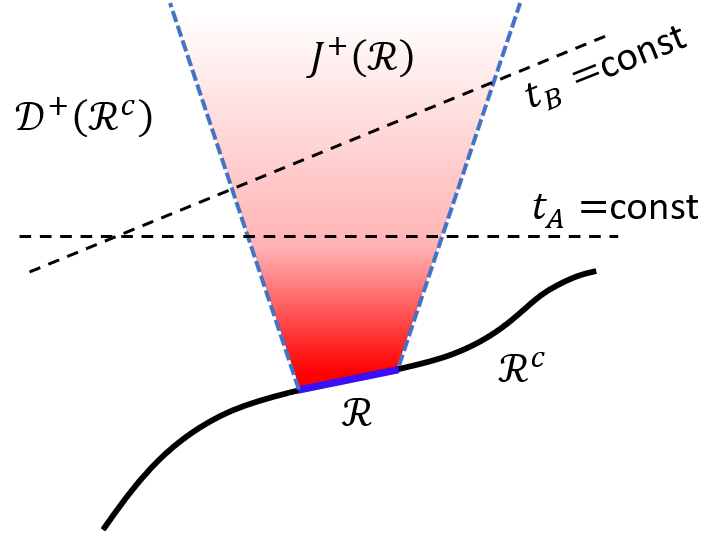}%
\includegraphics[width=0.48\textwidth]{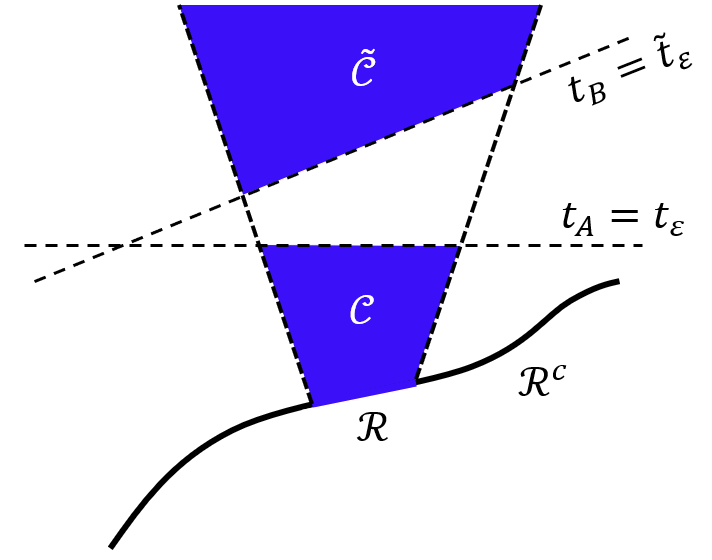}
	\caption{Left panel: Minkowski diagram of a sub-luminal perturbation (in Alice's coordinates). The blue segment is $\mathcal{R}$, where the perturbation is initially located, the black line is $\mathcal{R}^c$, where the perturbation is absent; together, $\mathcal{R}$ and $\mathcal{R}^c$ constitute the initial-data hypersurface $\Sigma$. The shaded red region is $J^+(\mathcal{R})$, where $\varphi$ can propagate. The white region above $\Sigma$ is $\mathcal{D}^+(\mathcal{R}^c)$, where $\varphi=0$. The shades of red are a color-map of $\varphi$ (red large, white small). The hyperplanes at constant $t_A$ and $t_B$ are respectively horizontal and oblique lines. Right panel: visualization of the sets $\mathcal{C}$ and $\tilde{\mathcal{C}}$ constructed in the proof of Theorem \ref{theo}.}
	\label{fig:fig2}
	\end{center}
\end{figure}

\begin{theorem}[Lorentz-invariance of dissipation]\label{theo}
If a sub-luminal perturbation is dissipated in the reference frame of Alice, it is also dissipated in the reference frame of Bob. 
\end{theorem}
\begin{proof}
Let's assume that the sub-luminal perturbation is dissipated in Alice's frame. Then, taken an arbitrary $\varepsilon >0$, we can find a time $t_\varepsilon$, future to $\mathcal{R}$, such that $\varphi < \varepsilon$ in $J^+(\mathcal{R}) \cap J^+_A (t_\varepsilon)$. Let $\mathcal{C}$ be the closure of $J^+(\mathcal{R}) \cap [J^+_A(t_\varepsilon)]^c$. Since $\mathcal{R}$ is bounded, $\mathcal{C}$ is compact (see figure \ref{fig:fig2}, right panel). On the other hand, $t_B$ is a continuous function; hence, also the image set $t_B(\mathcal{C}) \subset \mathbb{R}$ is compact. This implies that, fixed an arbitrary $\eta>0$, the real number 
\begin{equation}
\tilde{t}_\varepsilon := \eta + \max[t_B(\mathcal{C})]
\end{equation}
exists and is finite. Defined $\tilde{\mathcal{C}}:=J^+(\mathcal{R}) \cap J^+_B(\tilde{t}_\varepsilon)$, we have that $\mathcal{C} \cap \tilde{\mathcal{C}} = \emptyset$, because
\begin{equation}
\min[t_B(\tilde{\mathcal{C}})]= \tilde{t}_\varepsilon = \eta + \max[t_B(\mathcal{C})] > \max[t_B(\mathcal{C})] \, .
\end{equation}
Considering that, by definition, $\tilde{\mathcal{C}} \subset J^+(\mathcal{R}) \subseteq \mathcal{C} \cup [J^+(\mathcal{R}) \cap J^+_A (t_\varepsilon)]$, if follows that
\begin{equation}
 \tilde{\mathcal{C}} \subseteq J^+(\mathcal{R}) \cap J^+_A (t_\varepsilon) \, .
\end{equation}
However, if $ \tilde{\mathcal{C}} = J^+(\mathcal{R}) \cap J^+_B(\tilde{t}_\varepsilon)$ is a subset of $J^+(\mathcal{R}) \cap J^+_A (t_\varepsilon)$, then $\varphi <\varepsilon$ in $J^+(\mathcal{R}) \cap J^+_B(\tilde{t}_\varepsilon)$. 
\end{proof}
The essence of the proof can be easily understood by looking at the color-map in figure \ref{fig:fig2} (left panel): if the horizontal line $t_A = \text{const}$ is far enough in the future, the field $\varphi$ becomes arbitrarily small in the shaded region above it; then, we can always find an oblique line $t_B=\text{const}$ which slices $J^+(\mathcal{R})$ \textit{above} the horizontal line, as in the figure; in this way, we are sure that $\varphi$ is small also in Bob's frame, for a given time $t_B$ (and for later times). 

Figure \ref{fig:fig2} (left panel) also shows why the condition of sub-luminality is needed: the lines $t_A=\text{const}$ and $t_B=\text{const}$ always intersect somewhere; hence, an infinite portion of the line $t_B=\text{const}$ lies in the \textit{past} of $t_A=\text{const}$, where there is no bound on $\varphi$. Therefore, if $\varphi \rightarrow +\infty$ in the down-left corner of the figure (which is possible only if causality is violated), there is no limit on how large $\varphi$ can get in Bob's frame. This is exactly what happens in the argument of section \ref{thetought}. On the other hand, causality demands that $\varphi=0$ outside $J^+(\mathcal{R})$, so that, by pushing up the oblique line, we can make sure that $t_B=\text{const}$ is in the future of $t_A=\text{const}$ within the support of $\varphi$.

\subsection{Lorentz-invariance of linear instability}\label{SiNNuoz}

Theorem \ref{theo} deals with  non-linear perturbations, which are initially localised in space. However, in the linear approximation, it is usually convenient to study the evolution of sinusoidal plane waves, which have infinite support. Is there a straightforward analogue of Theorem \ref{theo} for sinusoidal plane waves?

We work with linear perturbations to a homogeneous stationary state, and call $\varphi:=\{ \delta \psi_i \}$ the array of perturbation fields $\delta \psi_i$. We take a global solution (i.e. a solution that is well defined across all Minkowski space-time) of the form
\begin{equation}\label{plainez}
\varphi = \text{``periodic field''} \times  e^{\Gamma_B t_B} \spc (\Gamma_B \in \mathbb{R}) \, ,
\end{equation}
where the periodic part is periodic both in space and in time. On hyperplanes $\{ t_B =\text{const} \}$, we have $\varphi=\text{``periodic field''} $, which implies that the perturbation may be a plane wave (i.e. a Fourier mode) in Bob's frame. This is the type of solution that one considers while performing a linear stability analysis in Bob's frame \cite{Hiscock_Insatibility_first_order,Kost2000}. Depending on the sign of $\Gamma_B$, the perturbation grows (if $\Gamma_B >0$), decays (if $\Gamma_B <0$), or has constant intensity (if $\Gamma_B=0$), in Bob's frame. Working in Alice's frame, $\varphi$ is no longer a Fourier mode (unless $\Gamma_B=0$, see Appendix \ref{appendixB1}), but it takes the form
\begin{equation}\label{varfuzzo}
\varphi = \text{``periodic field''} \times e^{\Gamma_B \gamma (t_A-vx_A)} \, .
\end{equation}
We can orient the $x_A$-axis in a way that $v>0$. Now, let us make two assumptions:
\begin{itemize}
\item The field equations are causal \cite{Hawking1973,Wald,BemficaCausality2018};
\item The perturbation grows in Bob's frame: $\Gamma_B >0$.
\end{itemize}
Our goal is to prove that the system is linearly unstable also in Alice's frame.

\begin{figure}
\begin{center}
\includegraphics[width=0.7\textwidth]{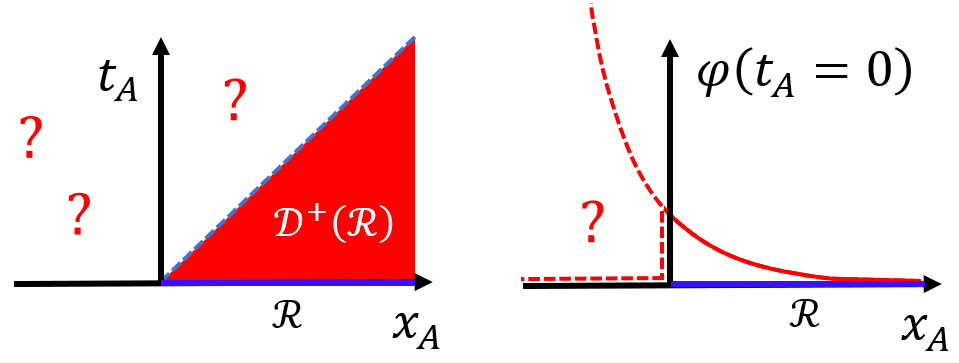}
	\caption{Left panel: an observer inside the red region $\mathcal{D}^+(\mathcal{R})$ cannot know what the initial state of the system was for $x_A<0$ (at $t_A=0$). Right panel: therefore, both $\varphi(t_A=0)$ and $\varphi^\star(t_A=0)=\Theta(x_A)\, \varphi(t_A=0)$ are initial states which are consistent with the data available to such observer; no experiment performed inside $\mathcal{D}^+(\mathcal{R})$ can tell $\varphi$ and $\varphi^\star$ apart. }
	\label{fig:figIgnor}
	\end{center}
\end{figure}

Consider an event $p \in \mathcal{D}^+(\mathcal{R})=\{t_A \geq 0 \} \cap \{x_A > t_A \}$, where $\mathcal{R}$ is the half-hyperplane (see figure \ref{fig:figIgnor})
\begin{equation}
\mathcal{R}:= \{ t_A=0 \} \cap \{ x_A >0 \} \, .
\end{equation}
By causality, $\varphi(p)$ cannot depend on the initial state of the system outside $\mathcal{R}$. In particular, if we consider an alternative solution $\varphi^\star$, whose initial data (for $t_A=0$) agrees with $\varphi$ on $\mathcal{R}$ and vanishes outside $\mathcal{R}$, i.e.
\begin{equation}
\varphi^\star(t_A=0)= \Theta(x_A) \, \varphi(t_A=0) \spc (\Theta = \text{Heaviside step function}) \, ,
\end{equation}
then we must have $\varphi^\star=\varphi$ on $\mathcal{D}^+(\mathcal{R})$. It follows that (for any $\varepsilon >0$, $t_A \geq 0$)
\begin{equation}
\varphi^\star\big|_{x_A=t_A+\varepsilon} \, = \, \varphi\big|_{x_A=t_A+\varepsilon} \, \propto \, e^{\Gamma_B \gamma (1-v)t_A}  \, \xrightarrow{t_A \rightarrow +\infty} \infty \, ,
\end{equation}
which means that both $\varphi$ and $\varphi^\star$ have divergent amplitude at future light-like infinity (see figure \ref{fig:fig3}). 
\begin{figure}
\begin{center}
\includegraphics[width=0.9\textwidth]{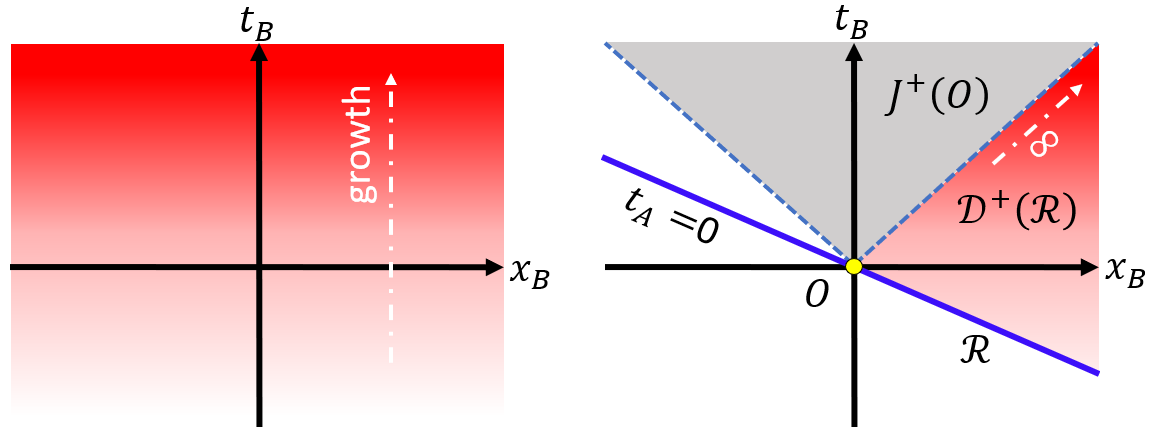}
	\caption{Minkowski diagram of the two solutions, $\varphi$ (left panel) and $\varphi^\star$ (right panel), in Bob's coordinates. The shades of red are a colormap of the the perturbation (the oscillatory behaviour of the periodic part is averaged out). On the grey area, we do not know the actual intensity of $\varphi^\star$. Left panel: $\varphi$ is an unstable Fourier mode (i.e. a growing plane wave) in B frame; it is well-defined across the whole space-time; its oscillation amplitude is constant along hyperplanes $ t_B = \text{const} $ (horizontal lines), and grows exponentially for growing $t_B$ ($\varphi \propto e^{\Gamma_B t_B}$). Right panel: $\varphi^\star$ is constructed on the half space-time $ \{ t_A \geq 0 \}$, by ``gluing'' initial data at $t_A=0$. On the right (on $\mathcal{R}$), we take $\varphi^\star(t_A=0)=\varphi(t_A=0)$, so that (by causality) $\varphi^\star=\varphi$ on $\mathcal{D}^+(\mathcal{R})$. On the left, we set $\varphi^\star(t_A=0)=0$ (hence $\varphi^\star=0$ on the respective Cauchy development). In this way, $\varphi^\star$ has a well-defined Fourier transform on $\{ t_A=0 \}$, but it diverges on $\mathcal{D}^+(\mathcal{R})$ (in the right-up corner), signalling an instability in Alice's frame. }
	\label{fig:fig3}
	\end{center}
\end{figure}
Now, it is not so surprising that $\varphi$ diverges somewhere in the future: in Alice's reference frame one has $\varphi(t_A=0) \propto \exp(-\Gamma_B \gamma v x_A)$, which is divergent at $x_A=-\infty$. Indeed, it is well-known that, if a perturbation has a divergent tail at $t_A=0$, its later exponential growth cannot be taken as an indication of instability of the field equations\footnote{For example, a perturbation of the form $\varphi=e^{t-x}$ is an exponentially growing solution ($\varphi \propto e^t$) of the causal wave equation $\nabla_a \nabla^a \varphi=0$ (that is obviously stable), with initial profile $\varphi(t=0)=e^{-x}$, which exhibits a divergent tail at $x=-\infty$.}. On the other hand, $\varphi^\star$ has a much more ``innocent'' initial state\footnote{For example, $\varphi^\star(t_A=0)$ has a well-defined Fourier transform. The reader should not be concerned about the discontinuity at $x_A=0$, because the step function can be replaced by any smooth function $\tilde{\Theta}$ such that $\tilde{\Theta}(x_A)=\Theta(x_A)$ for $x_A \in(-\infty, -1) \cup (0,+\infty)$, without affecting the result.}:
\begin{equation}
\varphi^\star(t_A=0) = \text{``periodic field''} \times  \Theta(x_A) \, e^{ -\Gamma_B \gamma v x_A} \, .
\end{equation} 
It is evident that, if such a perturbation diverges for later times, the system must be unstable in Alice's frame. We have, therefore, proven the following theorem:
\begin{theorem}[Lorentz-invariance of instability]\label{theo2}
If a causal (linear) theory presents a growing Fourier mode in one reference frame, then it is linearly unstable in all reference frames. 
\end{theorem}
Equivalently, if a causal theory is stable in one reference frame, there cannot be any growing Fourier mode in the boosted frames (analogue of Theorem 1 for plane waves). This results generalizes Theorem III of \citet{BemficaDNDefinitivo2020} to linear systems with arbitrary linear field equations. Theorem \ref{theo2} is also a generalization of the ``inverse argument'' of \citet{GavassinoCausality2021} to theories that do not have an entropy current with strictly non-negative divergence, such as DNMR \cite{Denicol2012Boltzmann} and BDNK \cite{Bemfica2019_conformal1}. Note that, for Theorem \ref{theo2} to hold, the unperturbed state does not need to be the state of global thermodynamic equilibrium; instead, it may just be a homogeneous and stationary background state. 

Let us, finally, give a less rigorous, but more intuitive, explanation of Theorem \ref{theo2}. Assume that, working in Alice's frame, we can split a given solution of the field equations into the product
\begin{equation}
\varphi = (\text{Intrinsic growth}) \times (\text{Drift}) = e^{\Gamma_A t_A} \times \varphi_D(x_A-w \, t_A) \, .
\end{equation}
Stability means $\Gamma_A <0$, causality requires $|w| \leq 1$. If we assume that $\varphi_D(x_A)=\text{``periodic field''} \times \exp(-\alpha x_A)$, with $\alpha>0$, we obtain
\begin{equation}\label{gagugo}
\varphi \propto e^{-\alpha x_A + (\Gamma_A+\alpha w)t_A} \, .
\end{equation}
Consistently with what we said before, we see that the fact that the perturbation grows in Alice's frame ($\Gamma_A +\alpha w>0$) does not necessarily mean that the theory is unstable ($\Gamma_A>0$), because a perturbation with an infinite tail (namely $\varphi =\infty$ at $x_A=-\infty$) can mimic an effective growth by drifting its tail. However, since $|w| \leq 1$, such effective growth cannot be too large in causal theories. Indeed, if we rewrite the perturbation \eqref{varfuzzo} in the form \eqref{gagugo}, we find that
\begin{equation}
\Gamma_A = \Gamma_B \gamma (1-vw) >0 \quad (\text{by causality}) \, ,
\end{equation}
signalling instability in Alice's frame. The reader can see the Appendix of \citet{GavassinoLyapunov_2020} for a similar argument.

\section{Acoustic-cone argument}\label{unishiipunti}

There is one ``global argument'', which unifies elegantly all the previous results, and gives a clear physical intuition of the underlying mechanism relating acausality and instability.

We start from a well-known fact: the outer characteristics  that pass through a space-time point $p$ bound the domain of influence of $p$ \cite{Susskind1969,Kost2000}. This implies that, if we perturb a system at $p$ (e.g., by coupling the field equations with an external source), the induced disturbance will be confined within a conical-like region called (future) acoustic cone\footnote{By ``acoustic cone'' we actually mean the outermost cone: the fastest characteristic. Here, we are using the evocative word ``acoustic'' to mean that observers inside it can ``feel'' the disturbance. But such disturbance does not need to be sound in a strict sense. It may also be a shear or an Alfv\'{e}n wave. Also, note that the acoustic cone is an actual 3D-cone (like the light-cone) only if the field equations are hyperbolic, and the medium is isotropic in some frame. For anisotropic media, the shape of the acoustic cone may be distorted. Furthermore, if the field equations are parabolic, the acoustic cone degenerates to a 3D-hyperplane. But the argument still applies.} \cite{kessence,DisconziAcoustic2019}. In addition, if the unperturbed state is a state of global thermodynamic equilibrium, and if the theory is dissipative, we can assume that the perturbation will be more intense at the tip of the cone (i.e. closer to $p$), and it will become smaller as we move far away from $p$.

\begin{figure}
\begin{center}
\includegraphics[width=0.6\textwidth]{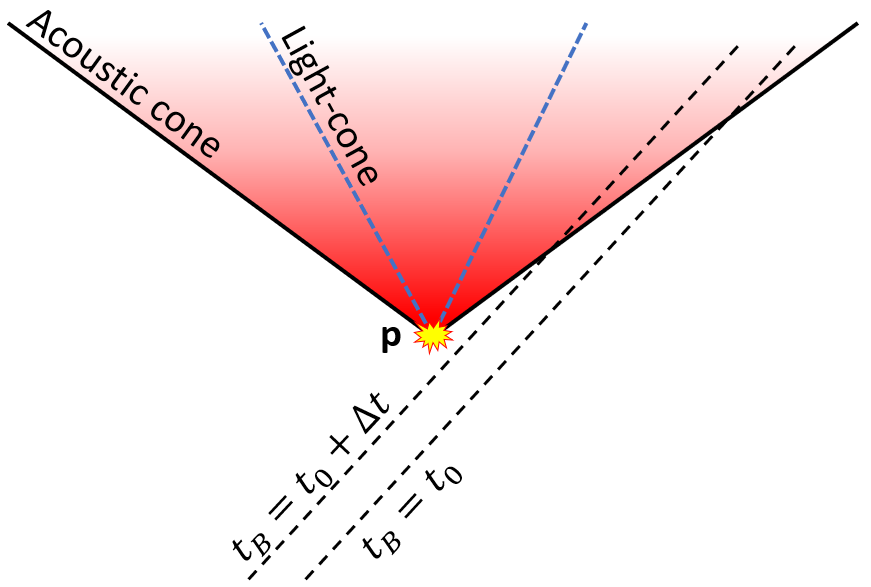}
	\caption{Minkowski diagram of the ``acoustic-cone argument''. An external source at $p$ (yellow star) generates a perturbation in the medium, which propagates within the outermost characteristic cone of the field equations (acoustic cone), and decays, by dissipation, as we move away from the source. If the field equations are acausal, the acoustic cone extends outside the light-cone. Hence, there is an observer Bob in whose frame the perturbation exists before $p$ has occurred. On the region $\{ t_B<t_B(p) \}$, Bob observes a solution of the source-less field equations, which is at equilibrium for $t_B  \ll t_B(p)$, but grows \textit{spontaneously} as $t_B$ approaches $t_B(p)$ from below.}
	\label{fig:fig55}
	\end{center}
\end{figure}

Let us first consider the case in which the theory is causal. Then, the acoustic cone is contained within (or overlaps) the light-cone. Therefore, \textit{all observers} experience the events in the following order: first $p$ (external source), then the tip of the cone (``intense perturbation''), then the rest of the cone (``damped perturbation''). Hence, all observers will agree that the equilibrium state is stable against perturbations. We have recovered Theorem 1 (at least qualitatively). Furthermore, if we assume that the source at $p$ excites all the Fourier modes, it is easy to recover Theorem 2.

Let us now move to the case in which the theory is acausal. In this case, a portion of the acoustic cone exits the light-cone. Thus, there is an observer (Bob) who measures the perturbation \textit{before} $p$ has occurred. In Bob's frame, as $t_B$ approaches $t_B(p)$ from below, the portion of the acoustic cone that intersects the hyperplane $\{ t_B = \text{const} \}$ gets closer to the tip of the cone, see figure \ref{fig:fig55}. This implies that: 
\begin{itemize}
\item for $t_B \ll t_B(p)$, the system is at equilibrium (the hyperplane $t_B = \text{const}$ is far from the tip of the cone);
\item for $t_B < t_B(p)$, the perturbation grows for increasing $t_B$;
\item at $t_B = t_B(p)$, the perturbation has a peak of intensity.
\end{itemize}
On the other hand, on the space-time region $\{ t_B<t_B(p) \}$, the perturbation is a solution of the field equations \textit{without} sources, because the only source is located at $p$. Therefore, we have shown that there is a solution of the source-less field equations, with initial data close to equilibrium [for $t_B \ll t_B(p)$], which departs from equilibrium at finite $t_B$ [just before $t_B(p)$]. This is a signature of instability, in Bob's frame. We have recovered the argument of section \ref{ilprimoluiluilui}: if the source of a perturbation can be delayed, then the system can spontaneously depart from equilibrium, in advance. But we have also recovered the argument of section  \ref{thetought}: just identify the wave-packet of figure \ref{fig:fig} with the front of the perturbation induced by $p$ (like a discontinuity, the front travels along the boundary of the acoustic cone \cite{Hishcock1983}).

At this point, we need to make a clarification. \citet{kessence} have suggested that, if the acoustic cone is larger than the light-cone, then one should just use the acoustic cone, in place of the light-cone, to define the causal structure of the space-time, and treat observers like Bob (figure \ref{fig:fig55}) as ``inappropriate'' observers, because they are not free to set the initial data at will. In this way, all paradoxes are avoided, and one has a new notion of causality. Their reasoning is valid, but we are working in different contexts. They are interested in what would happen in a universe in which there was some physical field which breaks the general-relativistic notion of causality at the \textit{fundamental level}: for them, the limitations of Bob are real. On the other hand, here we are assuming that general-relativistic causality is fundamentally valid in our Universe (hence, Bob is physically capable of shaping the system), but we are using a field theory that contradicts such principle. This is the actual origin of all paradoxes: not equations that break causality, but Cauchy problems that combine acausal theories with initial data on arbitrary space-like surfaces \cite{Kost2000}.

\subsection{Example: the boosted heat equation anti-diffuses!}\label{bheannn}

Using the ``acoustic-cone argument'' outlined above, we are finally able to show that the instability of the heat equation in moving reference frames \cite{Kost2000} is a consequence of its acausality. To this end, we consider the following thought experiment. A heat-conductive medium is at rest in Alice's frame. For $t_A<0$, the temperature is everywhere zero. At $t_A=0$, Alice injects a Dirac-delta of energy in the location $x_A=0$. For $t_A>0$, the spike of energy diffuses across the medium, according to the heat equation. The temperature field is therefore given by \cite{Morse1953}
\begin{equation}
T(t_A,x_A) = \dfrac{\Theta(t_A)}{\sqrt{4\pi D t_A}} \exp \bigg(  -\dfrac{x^2_A}{4Dt_A}\bigg) \, .
\end{equation} 
It can be easily verified (see \citet{Rauch_book}, section 1.7, Problem 3) that this function is indeed a $C^\infty$ solution of the heat equation for all values of $t_A$ and $x_A$, except at the point $p=(0,0)$, which is where the spike of energy is injected by Alice. Thus, when we boost to Bob's frame (treating $T$ as a scalar field \cite{MTW_book}),
\begin{equation}
T(t_B,x_B) =\dfrac{\Theta(t_B + v x_B)}{\sqrt{4\pi D \gamma (t_B + v x_B)}} \exp \bigg[  -\dfrac{\gamma (x_B + v t_B)^2}{4D(t_B + v x_B)}\bigg] \, , 
\end{equation}
and we restrict our attention to the spacetime region $\{ t_B<0 \}$, we obtain a $C^\infty$ solution of the boosted heat equation. In figure \ref{fig:green}, we show some snapshots of such solution.

\begin{figure}
\begin{center}
\includegraphics[width=0.5\textwidth]{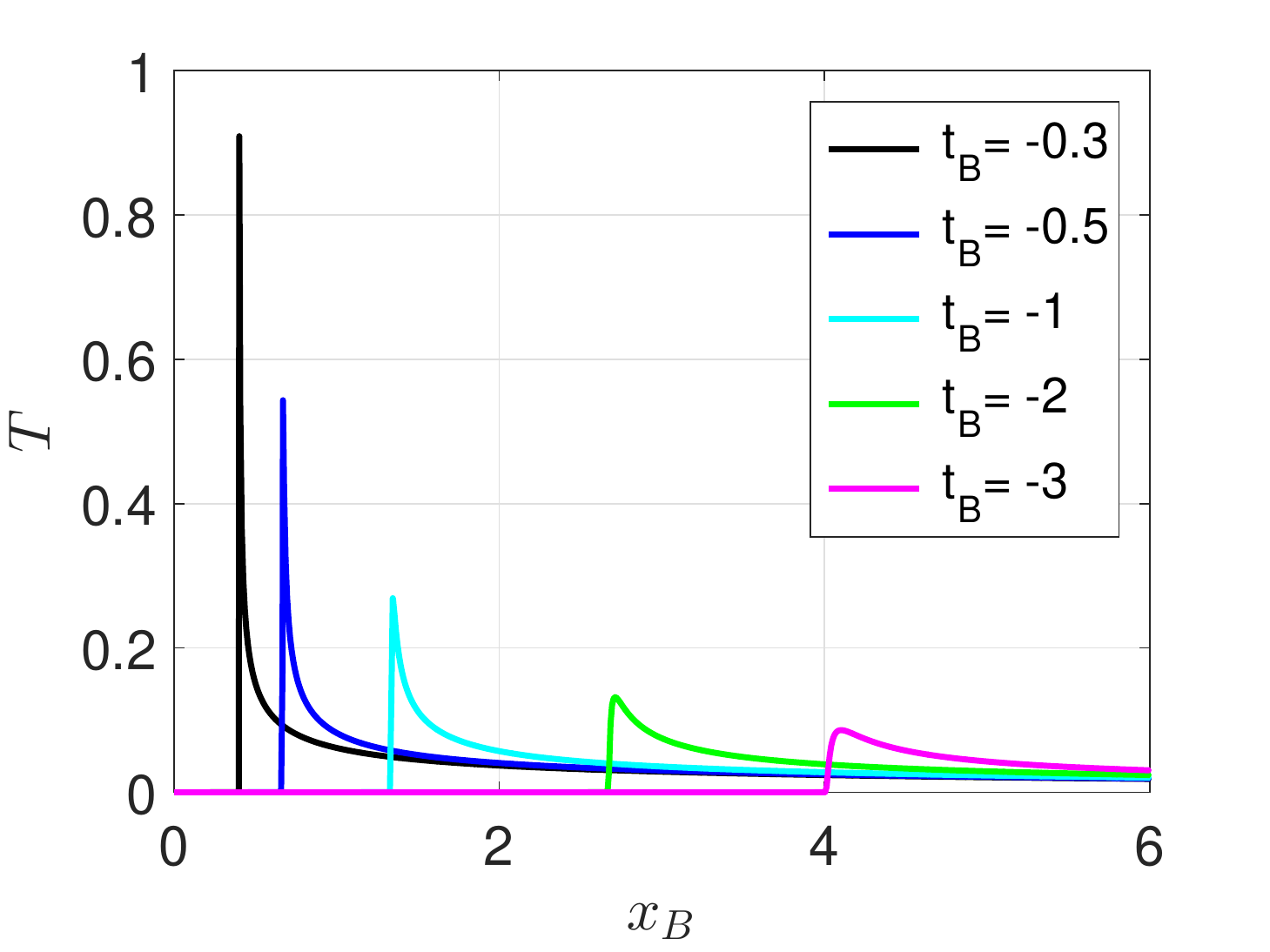}
	\caption{Boosted Green function of the heat equation, for $t_B<0$. We have set $D=30$ and $v=3/4$. Each curve represents a snapshot of $T(t_B,x_B)$, for different choices of $t_B$. If someone knows the entire history of the system, the interpretation of this figure is quite straightforward: Alice injects a spike of energy at $t_B=x_B=0$; because of acausality, a portion of such spike propagates towards the past; as it travels backward in time, the spike diffuses and flattens. On the other hand, to Bob (who cannot predict the decisions of Alice) the situation looks very different. From his perspective, the material is initially in thermodynamic equilibrium (at $t_B=-\infty$). Then, a perturbation builds up spontaneously, developing a superluminal front on the characteristic line $x_B=-t_B/v$. As time goes a head, the perturbation ``anti-diffuses'', becoming more and more peaked. Eventually, when $t_B \rightarrow 0$, the peak diverges at $x_B=0$. What we are observing is just an inversion of chronology (see figure \ref{fig:fig}).}
	\label{fig:green}
	\end{center}
\end{figure}

As we can see, the qualitative behaviour of $T(t_B,x_B)$ is consistent with our ``acoustic-cone argument''. Before Alice injects the spike, the temperature is already non-zero in Bob's frame: heat travels to the past! The characteristic line $x_B=-t_B/v$ (which is just the line $t_A=0$ expressed in Bob's coordinates) defines the ``acoustic cone'', and plays the role of a superluminal wave-front. There is a ``temperature wave'' on the right of such front, which is initially infinitesimal (for $t_B \ll 0$), and grows with time, ``anti-diffusing'', and becoming more and more peaked. In the end, $T$ develops a singularity at $t_B=0^-$. The very existence of a solution of this kind tells us that the boosted heat equation is ``anti-dissipative'' and unstable. 

But there is more. Let us focus on the infinite strip $ \{ t_B, x_B \} \in [-1,0) \times \mathbb{R}$. As we said, $T$ is $C^\infty$ on such strip. In addition, the right tail of $T$ decays faster than exponentially, while the left tail is identically zero. Therefore, we have constructed a solution of the boosted heat equation, whose initial data at $t_B =-1$ is regular (i.e. smooth and with well-defined Fourier transform), which nevertheless develops a singularity as $t_B \rightarrow 0$. It follows that the boosted heat equation must be ill-posed \cite{Kost2000}. This fact is not surprising. The boost has inverted the chronology of the heat equation (see subsection \ref{thetought}), converting it from diffusive to ``anti-diffusive''. Hence, the boosted heat equation should share some similarities with the ``backward heat equation'', $-\partial_t T = D \partial_x^2 T$, which is renowned for its ill-posedness.

\section{Some quick applications}\label{apliacia}

As we said in the introduction, a relativistic theory should pass three tests, to be considered reliable:
\begin{itemize}
\item[(i)] Causality,
\item[(ii)] Stability in the background's rest frame,
\item[(iii)] Stability in reference frames in which the background is moving.
\end{itemize}
Usually, one is content of verifying these properties at list for linear deviations from equilibrium, although in principle conditions (i,ii,iii) should be valid also in the non-linear regime.

The main message of this paper is that, once properties (i,ii) have been tested, assessing property (iii) is superfluous. In fact, if causality is violated, we know from the argument of section \ref{thetought} that the theory will be unstable (if dissipative). Furthermore, from the acoustic-cone argument of section \ref{unishiipunti}, we are also able to predict exactly in which reference frames the problems appear. If, on the other hand, (i,ii) are respected, then, by Theorems 1 and 2, (iii) follows automatically. Below we list some direct applications of the present results, which span all areas of relativistic physics, including heavy-ion collision simulations (point 1), accretion-disk simulations (point 2), alternative theories for dissipation (points 3-7), models for turbulent flow (point 8), Chern-Simons magnetohydrodynamics (point 9) and multi-constituent fluids (points 10-14).
\begin{enumerate}
\item  \citet{Plumberg2021} have shown that viscous heavy-ion collision simulations explore regimes of causality violation. This surely introduces uncertainty, but how much uncertainty? Each discrete time step in a simulation introduces error, and may ``activate'' Fourier modes. Picture this error as a small source on the right-hand side of the field equations. As shown in figure \ref{fig:fig55}, the effect of a source is dissipated away in those reference frames in which the acoustic cone points \textit{entirely} towards the future. However, in the remaining frames, it triggers growing modes. Hence, a simulation is really non-reliable if and only if a part of the acoustic cone ``sinks'' below the numerical time-step hyper-surfaces. Plotting the acoustic cone will, thus, show the real entity of the problem (the formula for the acoustic cone can be deduced from the causality analysis of \citet{BemficaPRL2021}).
\item \citet{Fragile2018} have performed relativistic viscous hydrodynamic simulations of accretion disks, adopting the \citet{landau6} theory, which is acausal: the acoustic cone is the normal hyperplane to the fluid's velocity \cite{Kost2000}. Thus, our reliability criterion  (see point 1) is violated at any point where the flow velocity is not normal to the 3+1 foliation: these simulation are probably non-reliable. However, the choice of approximating the viscous stress as constant (during the primitive solve) may have had the effect of erasing the second time-derivatives, effectively collapsing the acoustic cone upon the foliation, removing the pathologies. This would explain why some of their simulations predict the existence of stable disks, which is surprising, given the violence of the acausality-induced instabilities (see section \ref{thetought}). We believe that this issue needs further investigation. 
\item \citet{Pu2010} have shown that second-order viscous hydrodynamics is stable if and only if it is causal (in the linear regime). An analogous result has been found by \citet{BrutoThird2021} for third-order viscous hydrodynamics. We are in the position to predict that the same will also be true for higher-order viscous hydrodynamics. 
\item \citet{Lopez11} have formulated a relativistic theory for heat conduction, proving that it satisfies conditions (i,ii). Theorem 2 implies that also condition (iii) is satisfied: the theory is stable.
\item \citet{SrickerOttinger2019} have formulated a relativistic viscous theory for liquids. In \cite{SrickerOttinger2019}, they verify that, for some choice of parameters, condition (ii) is respected. However, we can see from figures 1,2,3 of \cite{SrickerOttinger2019} that, for this same choice of parameters, the front velocity of some Fourier modes is super-luminal. Since the signal velocity is not smaller than the front velocity \cite{Krotscheck1978}, we can conclude that the liquid under consideration violates causality and is, therefore, unstable in some reference frames.
\item \citet{VanStableFirst2012} have formulated a relativistic theory for viscosity and heat conduction, showing that it respects condition (ii). However, upon inspection of the last column of their matrix $\textbf{R}$ [equation (34)], we see that the field equations are not hyperbolic \cite{Hishcock1983}, suggesting the presence of causality violations and, thus, of instabilities. Indeed, if (in $\textbf{R}$) we impose $\Gamma=\gamma \tilde{\Gamma}$ and $k=i\gamma v \tilde{\Gamma}$ (spatially homogeneous solution in a boosted frame \cite{Hiscock_Insatibility_first_order}), we find that there is one growing solution for any $v \neq 0$.
\item  \citet{VanBiro2014} have formulated another theory for viscous hydrodynamics, similar to that discussed above. Unfortunately, it suffers exactly from the same problems as the previous one: the matrix $\textbf{R}$ [equation (38)] models acausal perturbations, which become unstable when boosted.
\item The Smagorinsky model \cite{Smagornki1963} is a filtered theory for modelling turbulent flows in large eddy Newtonian simulations.
\citet{Celora2021} have shown that, if the same approach is lifted to a relativistic setting, the resulting model is not  ``covariantly stable'', i.e. it satisfies condition (ii) but not condition (iii). Applying Theorem 2, we can conclude that the relativistic Smagorinsky model is acausal.
\item  \citet{Kiamari2021} have shown that Chern-Simons magnetohydrodynamics is causal, but unstable in the rest frame. Using Theorem 2, we can conclude that the theory must be unstable in every reference frame.
\item Many relativistic fluids can be modelled as reacting mixtures \cite{Burrows1986,BulkGavassino,Alford2020}. For a perfect-fluid reacting mixture, the rest-frame stability conditions coincide with the ``textbook'' conditions for thermodynamic stability \cite{GavassinoGibbs2021}, while the causality condition is simply the requirement that the sound-speed at frozen chemical fractions should not exceed the speed of light \cite{CamelioBulk1_2022}. Under these assumptions, by Theorem 2, a mixture is stable in all reference frames.
\item Most models for radiation hydrodynamics assume that there is a matter fluid with stress-energy tensor $M^{ab}$ and a radiation fluid with stress-energy tensor $R^{ab}$, which interact dissipatively though the equation $\nabla_a M^{ab}=-\nabla_a R^{ab}=G^b$, where $G^b$ is a hydrodynamic force \cite{Farris2008,Sadowski2013,GavassinoRadiazione}. Since $G^b$ usually does not depend on the gradients, its presence does not modify the characteristic determinant of the system. Therefore, the causality properties of the two fluids are unaffected by the coupling: if the dynamics of the matter fluid is acausal, the total radiation-hydrodynamic theory will also be acausal. On the other hand, radiation hydrodynamics is dissipative by construction \cite{Weinberg1971,GavassinoRadiazione}. Therefore, invoking the argument of section \ref{thetought}, we can conclude that all acausal fluids become unstable, when coupled with radiation through $G^b$.
\item The argument above can be easily generalised: assume that an arbitrary number of fluids and classical fields interact dissipatively through some equations $\nabla_a T^{ab}_{n}=G_{n}^b$ ($n$ is an index counting the fluids), where $G_{n}^b$ does not depend on the gradients. Then, if any of these fluids is acausal (and its dissipative coupling with the other fluids is not zero), the resulting composite system is unstable.
\item \citet{Carter_starting_point} have formulated a relativistic theory for superfluid mixtures. The simplest way of implementing dissipation in their theory is by coupling the currents through hydrodynamic forces which do not contain gradients \cite{GavassinoUEIT2021} (analogously to the case above). It follows that dissipative superfluid mixtures (and, more in general, ``multifluids'') are stable only if their non-dissipative analogue is causal. The only exception is when the dissipative coupling is mediated by quantum vortices \cite{langlois98,GavassinoIordanskii2021}, in which case the drag force depends non-linearly on the gradients, changing completely the causal structure of the system.
\item Superfluid neutron stars exhibit a phenomenon called ``entrainment'', according to which the superfluid momentum of the paired neutrons is not collinear to the flow of neutrons \cite{prix2004}. If we imagine to remove this effect, the acoustic cone becomes that of Carter's regular theory for heat conduction \cite{noto_rel}, which can be acausal, for certain equations of state \cite{OlsonRegularCarter1990}. Hence, the existence of the entrainment may be necessary to guarantee the stability of the equilibrium. The thermodynamic origin of this fact is studied in another work \cite{GavassinoStabilityCarter2022}.
\end{enumerate}

\section{Conclusions}


In this paper, we have identified the physical mechanism that connects causality, stability, and dissipation. Our reasoning can be summarised as follows. First, we have abstracted from the general notion of ``dissipation'' its key feature, namely the existence of a decaying-over-time scalar field (which measures ``how large'' a perturbation is at a point). Next, we have interpreted the word ``stability'' as the statement that all possible observers agree on the fact that such field is non-increasing with respect to their proper time. Finally, we have set up a simple argument: suppose that a perturbation moves superluminally (i.e., outside the light cone) and decays over time from the point of view of one observer. Because the perturbation is superluminal, it links causally disconnected space-time points which can, via a Lorentz transformation,
be chronologically inverted, making the decaying quantity appear increasing from the point of view
of another observer. In a nutshell, the lack of causality always allows one to transform dissipation
into ``anti-dissipation'' (i.e. dissipation backward in time). This also explains why acausal theories always turn out to be thermodynamically unstable \cite{GavassinoCausality2021}.

As a concrete example, we have studied how the retarded Green function of the heat equation transforms under Lorentz boosts. We have found that, due to relativity of simultaneity, one of  its Gaussian tails must always ``sink'' to the past (no matter how small the boost velocity), so that the boosted Green function presents an advanced part. This acausal precursor undergoes an inversion of chronology: it ``anti-diffuses'', instead of diffusing (see figure \ref{fig:green}). As a consequence, thermodynamics now is time-reversed: spikes tend to pinch (instead of flattening), energy tends to concentrate (instead of spreading), and the medium wants to move away from equilibrium (rather than towards it). That is why the boosted heat equation is unstable \cite{Hiscock_Insatibility_first_order}, anti-dissipative \cite{GavassinoUEIT2021}, and ill-posed \cite{Kost2000}.

With a similar reasoning, we have rigorously proved that, instead, if a \textit{causal} theory is stable in one reference frame, it is stable in all reference frames. The reason is that Lorentz transformations can never invert the chronological order of causally connected events: a decaying subluminal perturbation cannot be Lorentz-transformed into a growing one. In other words, causality guarantees that the ``thermodynamic arrow of time'' points towards the future in all reference frames, not only in the rest frame. 


Our analysis reveals that the causality-stability assessment is much easier than we thought, because the boosted-frame stability analysis (which is notoriously the most difficult part) is superfluous. Causality alone takes care of ensuring the Lorentz-invariance of a stability assessment, which can just be performed in a preferred reference frame. This result is a more general formulation of Theorem III of \citet{BemficaDNDefinitivo2020} and of the ``inverse argument'' of \citet{GavassinoCausality2021}. The main advantage of our Theorems 1 and 2 is that they do not make \textit{any assumption} about the structure of the field equations, besides causality.

We have also formulated a general criterion, based on the notion of ``acoustic cone'', which allows one to predict exactly in which reference frames an acausal theory becomes problematic. This criterion can be used to understand whether the reliability of state-of-the-art heavy-ion collision simulations \cite{Plumberg2021} is really compromised by the causality violations of Israel-Stewart-type theories. 


This paper has clarified several fundamental aspects of relativistic hydrodynamics and thermodynamics, providing a definitive answer to some old open questions:
\begin{itemize}
\item[1)] \textit{What is the ``physical interpretation'' of the instabilities that we observe in relativistic hydrodynamics?} They are just dissipative processes under time reversal. Without causality, there is no absolute notion of chronology, because the ``cause'' and the ``effect'' may be exchanged via a Lorentz boost. As a result, the ``thermodynamic arrow of time'' may point towards the past, for some observers. When this happens, systems evolve away from equilibrium, rather than towards it. That is why these instabilities are present in some reference frames and not in others.
\item[2)] \textit{Is it possible to make these instabilities small enough to be irrelevant?} No! If the beginning and the end of a process can be chronologically reordered via a boost, then there is an intermediate reference frame in which they are simultaneous. In such frame, the whole process occurs instantaneously. Therefore, one cannot hope that the instabilities will grow ``slowly'' (for a given acausal theory), because there is always some reference frame in which the growth rate is infinite.
\item[3)] \textit{Why does this problem appear only when we turn on dissipation?} Because non-dissipative theories are invariant under time reversal (strictly speaking, they are invariant under CPT \cite{weinbergQFT_1995}). Hence, in the absence of dissipation, an inversion of chronology does not produce any observable effect on the laws of thermodynamics.
 \item[4)] \textit{Is it possible to observe a similar phenomenon in Newtonian physics?} No. In Newtonian physics, time (and in particular chronology) is absolute. As a consequence, the thermodynamic arrow of time is Galilei-invariant, and all observers agree on whether a system is stable or not. 
\end{itemize}




Theorems \ref{theo} and \ref{theo2} are also interesting from the point of view of the foundations of relativistic thermodynamics. In fact, the essence of these theorems may be summarised as follows: \textit{if a system exhibits a tendency to evolve towards thermodynamic equilibrium in one frame of reference, it exhibits the same tendency in all frames, provided that the principle of causality holds}. This suggests that, once thermodynamics is valid in one reference frame, it should ``look the same'' in all reference frames. This is perfectly in line with our recent proof \cite{GavassinoLorentzInvariance2021} of van Kampen's argument \cite{vanKampen1968} for the existence of a relativistically covariant theory of thermodynamics. There, causality and stability were implicitly assumed when the concept of ``kick'' was introduced.


\section*{Acknowledgements}

This work was supported by the Polish National Science Centre grant OPUS 2019/33/B/ST9/00942. The author thanks M. Antonelli, B. Haskell and F. Bemfica for reading the manuscript and providing useful comments. I am particularly grateful to M. Disconzi, whose observations have allowed me to significantly improve the mathematical rigour of the discussion. My gratitude also goes to the editorial team and the referees of PRX: their recommendations and criticisms played a crucial role in giving this paper its final form.

\appendix

\section{The relativistic stability assessment}\label{AAAAAAAAA}

Let us compare the Galileian boost with the Lorentz boost (we ignore the variables $y$ and $z$):
\begin{equation}\label{TXxt}
\text{Galilei:} \quad
    \begin{cases}
      t_A=t_B \\
      x_A=x_B+v  t_B
    \end{cases}
    \spc
    \text{Lorentz:} \quad
        \begin{cases}
      t_A=\gamma (t_B + v x_B)\\
      x_A=\gamma (x_B + v t_B)
    \end{cases}
\end{equation}
Besides the Lorentz factor $\gamma=(1-v^2)^{-1/2}$, there is an additional term in the Lorentz boost that catches the eye: the position-dependent shift ``$\, vx_B \,$'' in the relativistic transformation of time. This term is responsible for a counter-intuitive phenomenon called ``relativity of simultaneity'', according to which two events that are simultaneous for one observer ($\Delta t_B=0$) may not be simultaneous for another observer ($\Delta t_A = \gamma v \Delta x_B \neq 0$). It is  this effect that makes the relativistic stability assessment more complicated than its Newtonian counterpart. Let us see why. 

\subsection{Boosted Fourier modes are no longer Fourier modes!}\label{appendixB1}

A physical system is in thermodynamic equilibrium. We perturb it a bit. We expect that, after an initial transient, the system will relax back to equilibrium. If, instead, the perturbation grows with time, we say that the theory is ``unstable''.

In practice, given a system of partial differential equations, how do we assess the stability of the equilibrium? The standard approach is the same both in Newtonian physics and in relativity, and works as follows. Let us say, for clarity, that we are interested in tracking the evolution of the local temperature $T(t_A,x_A)$, interpreted as a scalar field \cite{MTW_book}.  For small perturbations, we can work in the linear approximation, and expand a generic solution of the field equations as a superposition of sinusoidal plane-wave solutions (Fourier modes). For each of these solutions, the perturbation to the local temperature takes the form below:
\begin{equation}\label{TAAAAAZ}
\delta T (t_A , x_A) = e^{\Gamma_A t_A} \sin(k_A x_A -\omega_A t_A +\phi) \, ,
\end{equation}
where $\Gamma_A$, $k_A$, $\omega_A$, $\phi$ are real numbers, which do not depend on the spacetime location. The numbers $k_A$ and $\phi$, called respectively ``wavenumber'' and ``phase'', are treated as free parameters, whereas $\Gamma_A$ and $\omega_A$, called respectively ``growth rate'' and ``frequency'', are constrained by the equations of motion, and depend on $k_A$. It is evident that, if $\Gamma_A$ is always non-positive (for all values of $k_A$), the system is stable, otherwise it is unstable.

\begin{figure}
\begin{center}
\includegraphics[width=0.45\textwidth]{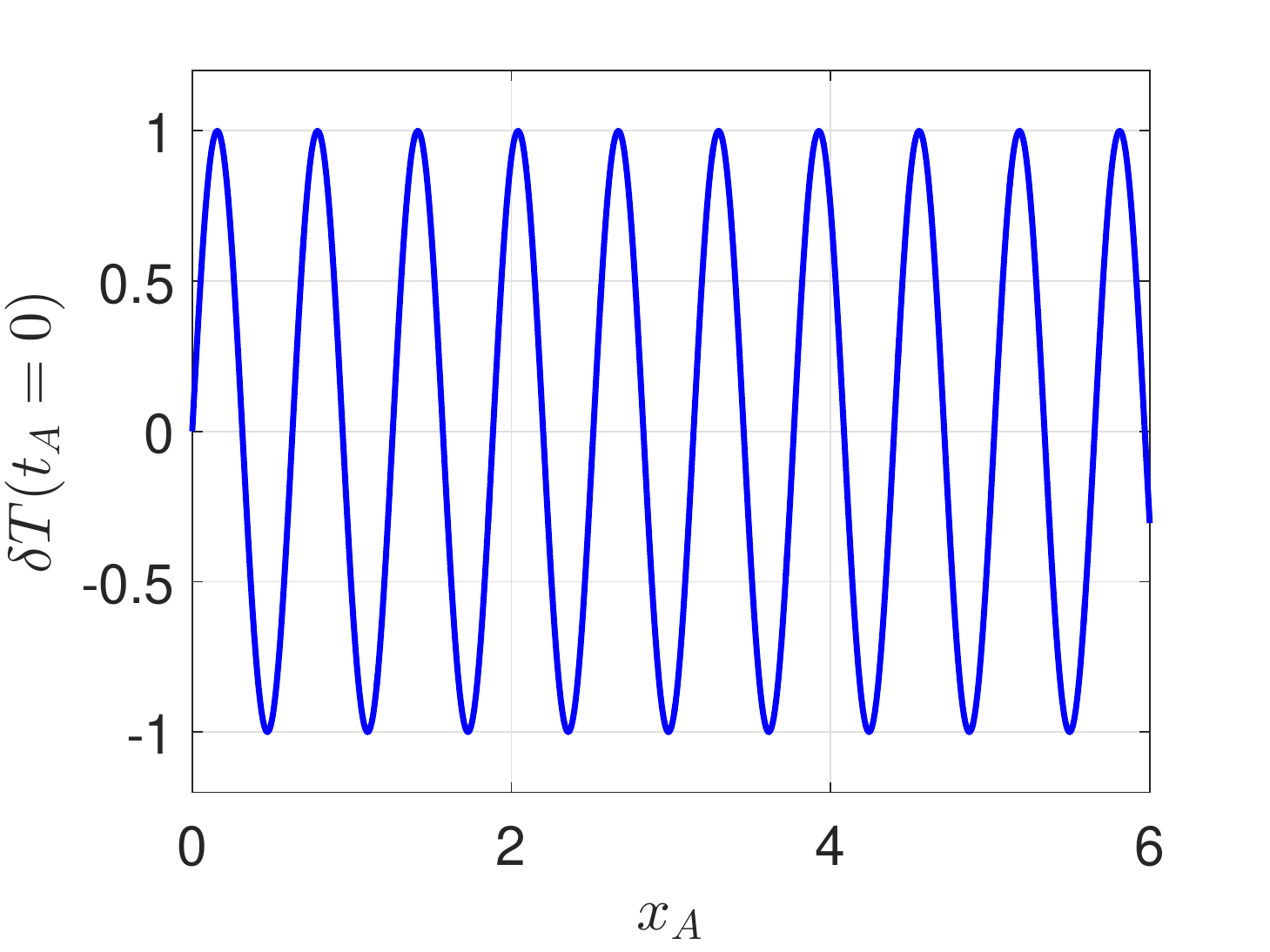}
\includegraphics[width=0.45\textwidth]{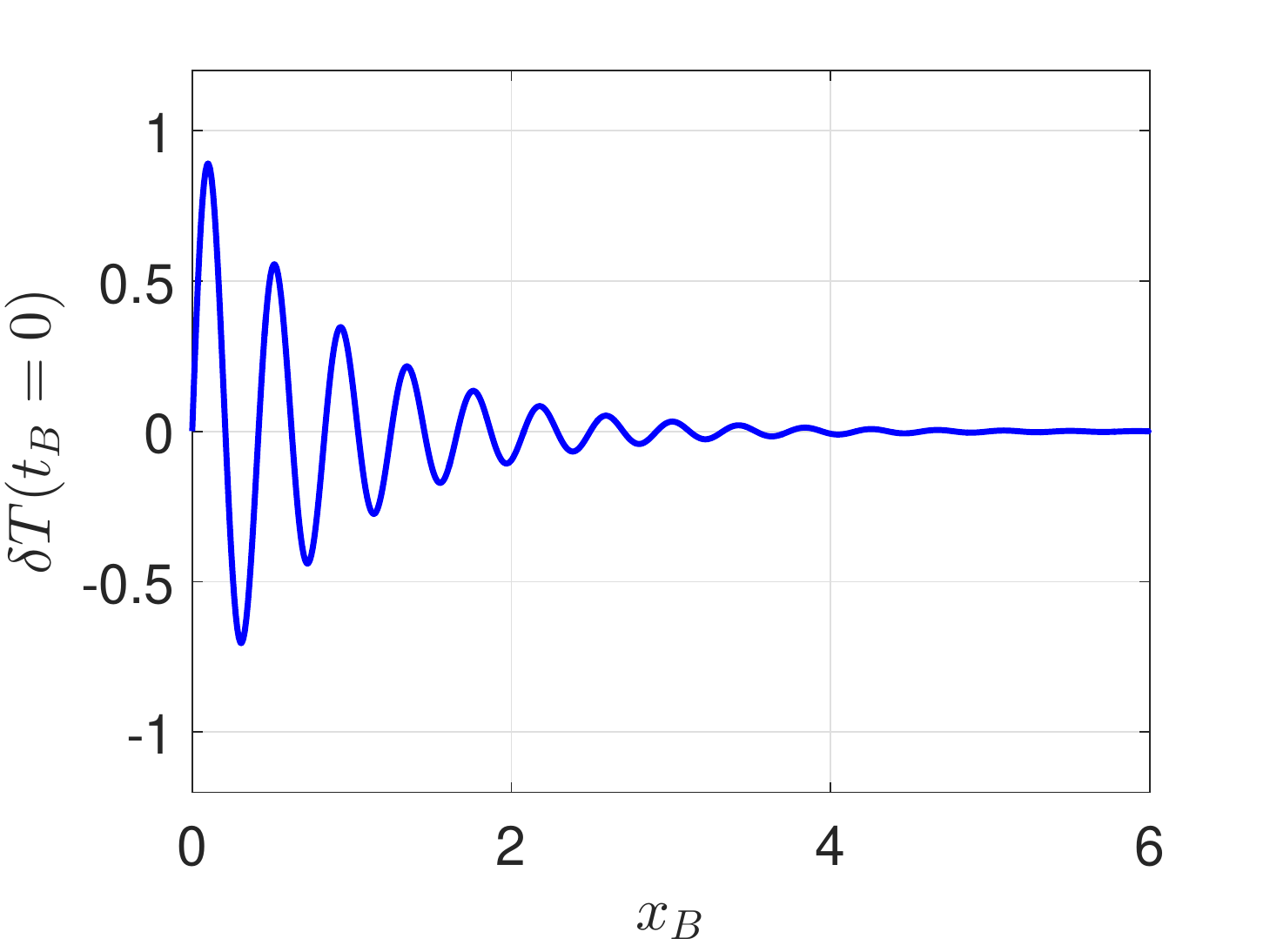}
	\caption{The same plane wave viewed by two observers in motion with respect to each other. The full spacetime dependence of the temperature perturbation is assumed to be $\delta T(t_A,x_A)=e^{-t_A}\sin(10 \,x_A)$. In the ``$\, A \,$'' frame, we have a conventional Fourier mode (left panel), whose amplitude decays in time. In the ``$\, B \,$'' frame, there is an exponential modulation also in space (right panel). This is a direct consequence of the relativity of simultaneity: if something decays only in time in one frame, it may decay both in time and space in another frame. We have chosen the boost velocity $v=3/4$, which corresponds to the Lorentz factor $\gamma \approx 1.5$.}
	\label{fig:Planes}
	\end{center}
\end{figure}

The only difference between a Newtonian stability analysis and a relativistic stability analysis lies in what happens when we change the frame of reference. Our intuition suggests that, once we have verified that the system is stable in one reference frame, it should be stable in all reference frames. And, indeed, this is true in Newtonian physics. In fact, when we change reference frame (in a Newtonian world), equation \eqref{TAAAAAZ} is transformed into
\begin{equation}\label{TBBBBBBZ}
\delta T (t_B , x_B) = e^{\Gamma_B t_B} \sin(k_B x_B -\omega_B t_B +\phi) \, ,
\end{equation}
with $\Gamma_B=\Gamma_A$, $k_B=k_A$, and $\omega_B=\omega_A-vk_A$. As we can see, the Galileian boost always maps sinusoidal plane waves into sinusoidal plane waves, with the same wavenumber and growth rate. Hence, if $\Gamma_A(k_A)$ cannot be positive, neither $\Gamma_B(k_B)$ can. However, things change dramatically in relativity. In fact, when we make a Lorentz boost, relativity of simultaneity mixes space with time, and the exponential in \eqref{TAAAAAZ} becomes
\begin{equation}\label{gammafavtfp}
e^{\Gamma_A t_A} = e^{\Gamma_A \gamma t_B} e^{\Gamma_A \gamma v x_B} \, .
\end{equation}
Because of the extra factor $e^{\Gamma_A \gamma v x_B}$, the wave is no longer sinusoidal in the boosted frame (see figure \ref{fig:Planes}), unless $\Gamma_A=0$. This is telling us that solutions of the form \eqref{TAAAAAZ} are intrinsically different from solutions of the form \eqref{TBBBBBBZ}. One is not the boosted version of the other! We cannot even express \eqref{TAAAAAZ}  as a superposition of solutions like \eqref{TBBBBBBZ}, because the factor $e^{\Gamma_A \gamma v x_B}$ has a divergent tail for $x_B \rightarrow \infty$ (plus or minus, depending on the sign of $v\Gamma_A$), so that the plane wave \eqref{TAAAAAZ} does not have a well-defined Fourier transform in the $B$ frame. This can lead to a surprising phenomenon: sometimes, a system is stable in one reference frame, but unstable in another one.

\subsection{The case of the heat equation}

The most striking example of how relativity of simultaneity can destabilize a system is the case of the heat equation:
\begin{equation}\label{btz}
\dfrac{\partial T}{\partial t_A} = D \dfrac{\partial^2 T}{\partial x_A^2} \, .
\end{equation}
In the ``$\, A \,$'' frame, this equation is clearly stable. In fact, if we plug \eqref{TAAAAAZ} into \eqref{btz}, we obtain $\Gamma_A=-Dk_A^2 \leq 0$. No Fourier mode can grow. However, quite surprisingly, there are unstable Fourier modes in all other frames of reference. For example, consider a solution of the form $\delta T(t_B)=e^{\Gamma_B t_B}$. In the ``$\, B \,$'' frame, this is a sinusoidal plane wave, with $k_B=0$. Physically, it models a configuration with no gradients in space for observer $B$. Intuitively, we would then expect that the only possible solution will be $\Gamma_B=0$ (no gradients $\Rightarrow$ no heat flux $\Rightarrow$ no temperature changes). However, this is not the case. Due to relativity of simultaneity, $\delta T$ acquires an exponential profile in the $A$ coordinates:
\begin{equation}\label{buconerozzo}
\delta T(t_A,x_A) =e^{\Gamma_B \, t_B} = e^{\Gamma_B \gamma (t_A-v x_A)} = e^{\Gamma_B \gamma t_A} e^{-\Gamma_B \gamma v x_A} \, .
\end{equation}
As a consequence, when we plug \eqref{buconerozzo} into \eqref{btz}, we obtain two possible solutions. One is $\Gamma_B = 0$. The other is
\begin{equation}
\Gamma_B = \dfrac{1}{D\gamma v^2} >0 \, .
\end{equation}
As we can see, in the $ B $ frame, the temperature is allowed to grow uniformly (with no bound), even in the absence of spatial gradients. Note that this is not in contradiction with the Fick law (``fluxes'' $\propto$ ``spatial gradients''), because in the the rest frame of the medium (the $A$ frame) there are gradients! In subsection \ref{bheannn} we will finally explain why the boosted heat equation \textit{must necessarily} be unstable.

\section{The Perturbation-Intensity field}\label{appendoxB}

The essence of Theorem 1 is the following: if a \textit{causal} perturbation with compact support converges to zero uniformly in Alice's frame, it converges to zero uniformly also in Bob's frame. To prove it, we only rely of the existence of a scalar field $\varphi$ that measures ``how large'' a perturbation is at a point. The simplest way of constructing $\varphi$ is the following. 
Suppose that Alice and Bob are interested in measuring a finite set of relevant scalar observables $\mathcal{O}_n$ (e.g temperature $T$, pressure $P$, chemical potential $\mu$, electromagnetic field-strength $F_{ab}F^{ab}$, etc...) at each spacetime event. Then, they may define $\varphi$ as
\begin{equation}
\varphi = \sum_n \big[ \mathcal{O}_n - \mathcal{O}_n^{\text{eq}}  \big]^2 \, ,
\end{equation}
where $\mathcal{O}_n^{\text{eq}}$ is the equilibrium value of $\mathcal{O}_n$. Viewed under this light, the theorem tells us that, if the differences $\mathcal{O}_n - \mathcal{O}_n^{\text{eq}}$ go below experimental resolution uniformly in Alice's frame, then the same happens in Bob's frame, provided that $\mathcal{O}_n = \mathcal{O}_n^{\text{eq}}$ on $\mathcal{D}^+(\mathcal{R}^c)$. 

On the other hand, one would like also to interpret the theorem in a more ``mathematical'' sense. For a given deterministic field theory, with some set of field equations, what are the underlying assumptions that make Theorem 1 applicable? We address this (rather technical) problem in the remaining part of this appendix. 

\subsection{What is a perturbation?}\label{techno}

We consider the following Cauchy problem:
\begin{equation}\label{Cauchyiuz}
\begin{cases}
   \mathcal{F}_h(\psi_i, \nabla_a \psi_i , \nabla_a \nabla_b \psi_i, ...)=0 \spc \text{on }\mathring{\mathcal{D}}^+(\Sigma) \\
  \psi_i=f^{(0)}_i, \, \, \, \,\, n^a\nabla_a \psi_i = f^{(1)}_i,... \spc \text{on }\Sigma \, ,
  \end{cases}
\end{equation}
where $\psi_i$ are the fields of the theory, $\Sigma$ is the space-like Cauchy 3D-surface introduced in the main text, $n^a$ is the unit normal to $\Sigma$, $\{ f^{(n)}_i \}_{i,n}$ is a set of functions on $\Sigma$ (they constitute the initial data), and $\mathcal{F}_h$ are some tensor-valued functions, which are smooth in all the arguments. We also assume that $\Sigma$ is smooth and we restrict our attention to smooth initial data. The following two assumptions are standard \cite{Wald}, but not so easy to guarantee in general\footnote{For example, we know that the Israel-Stewart theory is not globally well-posed \cite{DisonziISbraks2020}: a singularity may appear, in finite time, even for smooth initial data.}:
\begin{itemize}
\item The Cauchy problem \eqref{Cauchyiuz} is globally well-posed, i.e. the solution exists, is unique, and depends continuously on the initial data [across all $\mathcal{D}^+(\Sigma)$];
\item The field equations are causal, i.e. if the initial data for $\psi_i^\star$ agrees with that of $\psi_i$ on a subset $\mathcal{S}$ of $\Sigma$, then $\psi_i^\star=\psi_i$ on $\mathcal{D}^+(\mathcal{S})$.
\end{itemize}
Now we can define rigorously what we mean by a ``perturbation''. Any state of global thermodynamic equilibrium is modelled, in a deterministic field theory, as a specific solution of the field equations, with certain properties (e.g. $\nabla_a s^a=0$ and $\nabla_a \beta_b+\nabla_b \beta_a =0$ \cite{Israel_Stewart_1979}). We simply call $\psi_i$ a solution of this kind, which plays the role of the background equilibrium state. A localised perturbation (of the type considered in subsection \ref{sez3}) is an other solution $\psi_i^\star$, whose initial data agrees with that of $\psi_i$ on $\mathcal{R}^c$, but differs on $\mathcal{R}$ (there is no need for $\psi_i^\star$ to be ``close to $\psi_i$'', inside $\mathcal{R}$). Then, by causality, we know that
\begin{equation}\label{psiiprime}
\psi_i^\star = \psi_i  \spc \text{on }\mathcal{D}^+(\mathcal{R}^c).
\end{equation}
The final step consists of constructing a scalar field $\varphi$ which quantities how far $\psi_i^\star$ is from $\psi_i$ at a point. There are infinitely many ways of constructing such a field, but the simplest one works as follows: taken a preferred tetrad $e_A=e_A^a \partial_a$ (with $\nabla e_A=0$), and its dual $e^A=e^A_a dx^a$, introduce the operation
\begin{equation}
||A||_e^2 := \sum_{A_1,...,A_l,B_1,...,B_k}^{\text{from }0 \text{ to }3} |A(e^{A_1},...,e^{A_l},e_{B_1},...,e_{B_k})|^2 \, , 
\end{equation}
for a generic complex-valued $(l,k)$-tensor $A$. Then, define $\varphi$ as
\begin{equation}\label{piffuz}
\varphi =\sum_{i}  ||\psi_i^\star- \psi_i||_e^2 \, .
\end{equation}
It is evident that, given a space-time point $p$, $\varphi(p)=0$ if and only if $\psi_i^\star(p)=\psi_i(p)$. Hence, from equation \eqref{psiiprime}, we see that $\varphi=0$ on $\mathcal{D}^+(\mathcal{R}^c)$, proving that Definition 1 follows directly from causality. Furthermore, global well-posedness guarantees that $\varphi$ exists everywhere in $\mathcal{D}^+(\Sigma)$, which is another central assumption of the theorem.

\bibliography{Biblio}

\newpage

\section*{Supplementary material}

We test the predictions of the argument of section III.A (main text) for the case of the telegraph equation, showing that the resulting formula for the growth rate of the perturbation coincides with that computed with the Fourier analysis. We also verify explicitly that, in those reference frames in which the perturbation grows with time, the entropy grows with the perturbation. The calculations are performed in $1{+}1$ dimensions.

\maketitle

\subsection*{Implications of the argument}

The local temperature can be interpreted as a scalar field \cite{MTW_book}, provided that we define it as $T:=(\beta^a \beta_a)^{-1/2}$, where $\beta^a$ is the inverse-temperature four-vector \cite{Israel_2009_inbook,Becattini2016,GavassinoTermometri2020}.
Let's assume that, in Alice's rest frame, the scalar field $T$ is governed by the telegraph equation
\begin{equation}\label{fieldEquaz}
\dfrac{\partial_t^2 T}{w^2}+\dfrac{\partial_t T}{D} = \partial_x^2 T \, ,
\end{equation}
where $w>0$ and $D>0$ are some constant coefficients. For $w=+\infty$ we recover the heat equation, while for $D=+\infty$ we have a non-dissipative wave equation. Let us estimate the stability properties of equation \eqref{fieldEquaz}, just considering the physical setting outlined in section III.A of the main text.

The characteristic surfaces of \eqref{fieldEquaz} are
\begin{equation}
x = \pm w \, t + \text{const}.
\end{equation}
This implies that the fastest wave-packets allowed by the theory travel with speed $w$. Indeed, in the limit of highly-oscillating wave-packet (i.e., for infinitely large gradients), equation \eqref{fieldEquaz} becomes
\begin{equation}
(\partial_t^2  - w^2 \partial^2_x) T \approx 0 \, ,
\end{equation}
which is a wave equation with characteristic speed $w$. This suggests us that, to model a localised high-frequency wave-packet, we may consider the (approximate) ansatz solution\footnote{Since the perturbation travels with speed $w$, the coefficient $w$ considered here coincides with the factor $w$ introduced in the main text.}
\begin{equation}\label{Txx}
T(x,t) \approx e^{\Gamma t} \, T_0(x-wt) + \text{const} \spc (T_0 \text{ has compact support and is highly-oscillating}) ,
\end{equation}
where the exponential factor models the damping effect induced by the dissipative term $\partial_t T/D$ in \eqref{fieldEquaz}. Plugging this ansatz solution into \eqref{fieldEquaz}, and working in the limit of large gradients, we obtain a formula for $\Gamma$:
\begin{equation}
\Gamma = -\dfrac{w^2}{2D} \, .
\end{equation}
As expected, $\Gamma <0$, meaning that the perturbation is damped, in the reference frame of Alice. Finally, if we make the formal identification $\varphi \equiv \text{exp}(\Gamma t)$, we can use equation (5) of the main text to compute the growth-rate of the perturbation in the reference frame of Bob:
\begin{equation}\label{gammab}
\Gamma_B := \dfrac{d \ln \varphi}{dt_B} =\dfrac{1}{\gamma(1-vw)} \dfrac{d \ln \varphi}{dt}= -\dfrac{w^2}{2D\gamma (1-vw)} \, .
\end{equation}
Upon examination of equation \eqref{gammab}, we can conclude that:
\begin{itemize}
\item The wave equation ($D=+\infty$) is stable ($\Gamma_B =0$);
\item The telegraph equation ($w$ and $D$ finite) is unstable in Bob's frame if\footnote{Obviously, it is unstable also for $v<-w^{-1}$. To see this, one just needs to replace the right-travelling solution \eqref{Txx}, with the left-travelling solution $\text{exp}(\Gamma t) \, T_0(x+wt)$.} $v>w^{-1}$;
\item The heat equation ($w\rightarrow +\infty$) is unstable in Bob's frame for any $v \neq 0$.
\end{itemize}

\subsection*{Consistency with the Fourier analysis}

Let us, now, verify that equation \eqref{gammab} is consistent with the value of $\Gamma_B$ that one obtains computing the boosted dispersion relations directly.

If we work in Alice's frame, and consider a solution of the form
\begin{equation}
T(x,t)=T(0,0) \, e^{ikx-i\omega t} \, ,
\end{equation} 
equation \eqref{fieldEquaz} becomes
\begin{equation}\label{pollino}
\dfrac{\omega^2}{w^2} + \dfrac{i \omega}{D} - k^2 =0 \, .
\end{equation}
Recalling that $(\omega,k)$ transforms as a vector, we can relate $\omega$ and $k$ to the frequency $\tilde{\omega}$ and wave-vector $\tilde{k}$ in Bob's frame:
\begin{equation}\label{changuz}
\omega = \gamma (\tilde{\omega}+v \, \tilde{k}) \spc \spc k = \gamma (\tilde{k}+v \, \tilde{\omega}) \, .
\end{equation}
Plugging \eqref{changuz} into \eqref{pollino} we obtain
\begin{equation}
(1-v^2 w^2) \, \tilde{\omega}^2 + \bigg[ 2v(1-w^2)\tilde{k} + \dfrac{iw^2}{\gamma D}  \bigg] \tilde{\omega} + (v^2-w^2)\tilde{k}^2 + \dfrac{ivw^2}{\gamma D} \tilde{k}=0 \, .
\end{equation}
This produces two dispersion relations:
\begin{equation}\label{ilfinale}
(1-v^2 w^2) \, \tilde{\omega}_{\pm}(\tilde{k}) = -v(1-w^2)\tilde{k}-\dfrac{iw^2}{2\gamma D} \pm \dfrac{w}{\gamma} \sqrt{\dfrac{\tilde{k}^2}{\gamma^2} - \dfrac{iw^2}{\gamma D} v \tilde{k} - \dfrac{w^2}{4D^2}} \, .
\end{equation}
Recalling that the wave-packet \eqref{Txx} is a high-frequency solution, we can take the limit of large $\tilde{k}$. This allows us the expand the square root in \eqref{ilfinale}, leading to the following result:
\begin{equation}
 \tilde{\omega}_+ (\tilde{k}) = \dfrac{w-v}{1-vw} \, \tilde{k} - \dfrac{iw^2}{2 D \gamma (1-vw)}   \spc \spc  \tilde{\omega}_-  (\tilde{k})= \dfrac{-w-v}{1+vw} \, \tilde{k} - \dfrac{iw^2}{2 D \gamma (1+vw)} \, .  
\end{equation}
It is evident, from the real part of the dispersion relation, that $\tilde{\omega}_+$ describes a perturbation that in Alice's rest frame is drifting with velocity $+w$ (recall the Lorentz transformation of velocities). Analogously, $\tilde{\omega}_-$ describes a perturbation that in Alice's rest frame is drifting with velocity $-w$. If follows that the growth rate of the wave-packet \eqref{Txx}, as measured in Bob's frame, is
\begin{equation}
\Gamma_B = \text{Im} \, \tilde{\omega}_+ = - \dfrac{w^2}{2 D \gamma (1-vw)}  \, .
\end{equation}
This shows that the growth rate $\Gamma_B$ predicted using the argument of section III.A (main text) coincides with that extracted from the explicit Fourier analysis.

\subsection*{Violation of the maximum entropy principle}

Let us study the sign of the entropy perturbation carried by the wave-packet.

In order to perform a thermodynamic analysis, we need to know the constitutive relations of all the Noether currents of the system, plus the constitutive relation of the entropy current $s^a$ \cite{GavassinoLyapunov_2020,GavassinoGibbs2021,GavassinoCausality2021}. Here, to keep the discussion simple, we will make the assumption that there is only one relevant conservation law, with associated current $J^a$. Hence, imposing that the thermodynamic state of the medium can be completely characterised using only the scalar field $T$, and an auxiliary field $q$ (the heat flux), we postulate the following non-equilibrium constitutive theory\footnote{For simplicity, we are also assuming that the equilibrium equation of state of the material is $S=C_v \ln T$, so that the equilibrium parts of $s^0$ and $J^0$ are respectively $c_v \ln T$ and $c_v T$.} \cite{GavassinoUEIT2021}:
\begin{equation}\label{Jasa}
J^a = 
\begin{pmatrix}
c_v T   \\
q \\
\end{pmatrix}
\spc 
s^a =
\begin{pmatrix}
c_v \ln T -Bq^2/2  \\
q/T \\
\end{pmatrix}
\spc \text{with} \quad c_v,B = \text{const} >0 \, .
\end{equation}
The dynamics is then completely determined by the rate-equations
\begin{equation}\label{Fondamentali}
\nabla_a J^a =0  \spc \nabla_a s^a = \dfrac{q^2}{\kappa T^2} \spc \text{with} \quad \kappa=\text{const} >0 \, ,
\end{equation}
which explicitly read
\begin{equation}\label{buonalaprima}
c_v \partial_t T +\partial_x q =0  \spc BT^2 \partial_t q +\dfrac{q}{\kappa} +\partial_x T =0 \, .
\end{equation}
If we linearise these two equations around an equilibrium state (namely, a state with $T=\text{const}$ and $q=0$), we can combine them together to recover equation \eqref{fieldEquaz}, with
\begin{equation}\label{dsb5}
w^2 = \dfrac{1}{c_v B T^2}  \spc D = \dfrac{\kappa}{c_v} \, .
\end{equation}
Hence, the telegraph equation \eqref{fieldEquaz} is the natural dynamical equation of a thermodynamic system having the constitutive relations \eqref{Jasa}. 

Now that we have assigned some constitutive relations to the system, we can proceed to compute the entropy variation. First of all, let us compute the perturbation to the entropy current around the equilibrium state with temperature $T$ (truncating the expansion to second order in $\delta T$ and $\delta q$):
\begin{equation}\label{saaaa}
\delta s^a = \dfrac{\delta J^a}{T} - \dfrac{1}{2T^2}
\begin{pmatrix}
c_v (\delta T)^2 + BT^2 (\delta q)^2   \\
2\delta q \delta T \\
\end{pmatrix} + (\text{Third order terms}) \, .
\end{equation}
Secondly, we can use \eqref{buonalaprima} to show that, for a high-frequency wave-packet of the form \eqref{Txx}, one has
\begin{equation}\label{qaaa}
\delta q \approx wc_v \delta T \, .
\end{equation}
Finally, we can take the flux of \eqref{saaaa} across Bob's surface of contemporary events ($t_B = \text{const}$) to obtain the perturbation to the total entropy $S_B$ (as measured in Bob's frame). The result is
\begin{equation}\label{sbuz}
\delta S_B \approx \dfrac{\delta U}{T} - (1-vw)\dfrac{\gamma c_v}{T^2} \int_{t_B=\text{const}} \! \! \! \! \! \! \! \!\!  (\delta T)^2 \, dx_B \, ,
\end{equation} 
where the quantity $U$ is the conserved scalar charge associated with the current $J^a$. Now we see the problem: the entropy should be maximised at equilibrium, for a fixed value of the integral of motion $U$; in other words, we should have
\begin{equation}
\delta S_B \leq 0  \spc \text{as long as} \spc \delta U=0 \, .
\end{equation}
However, for $v > w^{-1}$, the second term on the right-hand side of \eqref{sbuz} becomes positive, meaning that the maximum entropy principle is violated in Bob's frame. Recall that we can have $v>w^{-1}$ only if $w>1$, i.e. if causality is violated.

We can finally compute the growth rate of the perturbation directly from thermodynamic considerations. Let us, first of all, define the integral
\begin{equation}
I(t_B) := \int_{\text{Bob's time }t_B}  \! \! \! \!\! \! \! \! \! \! \! \! \!\! \! \! \! \! \! \! \! \!\!  (\delta T)^2 \, dx_B = I(0) \, e^{2\Gamma_B t_B} \, ,
\end{equation}
where $\exp(\Gamma_B t_B)$ is the growth factor of $\delta T$ in Bob's frame. Taking the time-derivative of \eqref{sbuz} (recalling that $U$ is a constant of motion), we obtain
\begin{equation}\label{dsb1}
\dfrac{dS_B}{dt_B} =  - 2(1-vw)\dfrac{\gamma c_v}{T^2} I(t_B) \,  \Gamma_B \, .
\end{equation}
On the other hand, we can also compute the time-derivative of the entropy using the second equation of \eqref{Fondamentali}, together with equation \eqref{qaaa}: 
\begin{equation}\label{dsb2}
\dfrac{dS_B}{dt_B} = \dfrac{1}{\kappa T^2} \int_{\text{Bob's time }t_B}  \! \! \! \!\! \! \! \! \! \! \! \! \!\! \! \! \! \! \! \! \! \!\!  (\delta q)^2 \, dx_B = \dfrac{w^2 c_v^2 }{\kappa T^2} \,  I(t_B) \, .
\end{equation}
Comparing \eqref{dsb1} with \eqref{dsb2}, and using \eqref{dsb5}, we recover equation \eqref{gammab}, namely
\begin{equation}
\Gamma_B = -\dfrac{w^2}{2D\gamma (1-vw)} \, .
\end{equation}

\label{lastpage}

\end{document}